\documentclass[a4paper]{amsart}

\usepackage{amsmath,amssymb,amsthm,color,hyperref}
\usepackage[T1]{fontenc}

\newtheorem{corollary}{Corollary}
\newtheorem{lemma}{Lemma}
\newtheorem{proposition}{Proposition}
\newtheorem{remark}{Remark}
\newtheorem{theorem}{Theorem}

\def\C{\mathbb C}
\def\N{\mathbb N}
\def\R{\mathbb R}
\def\e{{\rm e}}
\def\Id{{\rm Id}}
\def\Im{{\rm Im}}
\def\Re{{\rm Re}}
\def\op{{\rm op}}
\def\tr{{\rm tr}}
\def\eps{\varepsilon}
\def\phi{\varphi}
\def\F{{\mathcal F}}
\def\H{{\mathcal H}}
\def\K{{\mathcal K}}

\def\T{{\mathcal T}}
\def\W{{\mathcal W}}

\DeclareMathOperator{\artanh}{artanh}

\title{Hagedorn wavepackets in time-frequency and phase space}
\author[C. Lasser]{Caroline Lasser}
\thanks{This research was supported by the German Research Foundation
(DFG), Collaborative Research Center SFB-TR 109.}
\address{Zentrum Mathematik, Technische Universit\"at M\"unchen, 80290 M\"unchen, Germany}
\email{classer@ma.tum.de}

\author[S. Troppmann]{Stephanie Troppmann}
\address{Zentrum Mathematik, Technische Universit\"at M\"unchen, 80290 M\"unchen, Germany}
\email{steffi.troppmann@mytum.de}

\date{\today}


\begin{document}

\begin{abstract}
The Hermite functions are an orthonormalbasis of the space of square integrable functions with favourable approximation properties. Allowing for a flexible localization in position and momentum, the Hagedorn wavepackets generalize the Hermite functions also to several dimensions.
Using Hagedorn's raising and lowering operators, we derive explicit formulas and recurrence relations for the 
Wigner and FBI transform of the wavepackets and show their relation to the Laguerre polyomials. 
\end{abstract}

\maketitle

\section{Introduction}

The Hermite functions are an important member of the family of special functions. They form an orthonormal basis of the space of square integrable functions on the real line. 
The Hermite functions are eigenfunctions of the Fourier transform and of the harmonic oscillator. They can be generated either by raising and lowering operators or by a three-term recurrence relation. 
And one could name more of their distinguished properties. 


The Hagedorn wavepackets \cite{H81,H85,H98} generalize the Hermite functions to several space dimensions while adding more flexibiliy in terms of position and momentum localization. 
Moreover, they allow for a positive scale parameter $\eps>0$ setting the wavepackets' width at the order of $\sqrt\eps$ and their wavelength at the order of $\eps$. In particular, they are constructed from the complex Gaussian function
$$
\phi^\eps_0(x) = (\pi\eps)^{-d/4} \det(Q)^{-1/2} \exp(\tfrac{i}{2\eps}(x-q)^T PQ^{-1}(x-q)+\tfrac{i}{\eps}p^T(x-q))
$$
centered in position $q\in\R^d$ and momentum $p\in\R^d$, where the complex matrices $Q,P\in\C^{d\times d}$ satisfy the symplecticity condition
$$
Q^T P - P^T Q = 0,\qquad Q^* P - P^* Q = 2i\Id.
$$
That is, the real matrix
$$
F = \begin{pmatrix}\Re(Q)& \Im(Q)\\ \Re(P)& \Im(P)\end{pmatrix}\in\R^{2d\times 2d}
$$
is symplectic.
In this way, the Gaussian has a position density $|\phi^\eps_0(x)|^2$, which is proportional to the multivariate normal distribution with mean $q$ and covariance matrix $\tfrac{\eps}{2}QQ^*$. 
The corresponding momentum distribution $|\F^\eps\phi^\eps(\xi)|^2$, defined via the $\eps$-scaled Fourier transform
$$
\F^\eps\phi(\xi) = (2\pi\eps)^{-d/2} \int_{\R^d} \phi(x) \e^{-ix^T\xi/\eps} dx,\qquad \xi\in\R^d,
$$
is proportional to the normal distribution with mean $p$ and covariance $\tfrac{\eps}{2}PP^*$. The familiar raising operator of the Hermite functions is generalized to 
$$
A^\dagger = \tfrac{i}{\sqrt{2\eps}} \left( P^*(x-q) - Q^*(-i\eps\nabla_x - p)\right),
$$
and we obtain the $k$th Hagedorn wavepacket $\phi^\eps_k$ by the $k$fold application of the raising operator to the initial Gaussian, that is,  
$$
\phi^\eps_k = \tfrac{1}{\sqrt{k!}}(A^\dagger)^k \phi^\eps_0,\qquad k\in\N^d. 
$$
The $k$th Hagedorn wavepacket $\phi^\eps_k$ is the product of a multivariate polynomial $p^\eps_k$ of degree $|k|=k_1+\cdots+k_d$ with the complex Gaussian $\phi^\eps_0$,
$$
\phi^\eps_k(x) = \frac{1}{\sqrt{2^{|k|}k!}}\, p^\eps_k(x) \phi^\eps_0(x),\qquad x\in\R^d.
$$ 
In the specific case of a real symmetric matrix $Q=Q^T\in\R^{d\times d}$, the polynomials~$p^\eps_k$ can be expressed as a product of $d$ univariate Hermite polynomials. Otherwise, the situation is more complicated, and it is one of our aims to establish properties of classical univariate orthogonal polynomials like sum rules and a Rodriguez formula for the multivariate polynomials $p^\eps_k$.

The complex Gaussian $\phi^\eps_0$ has various names in the literature. Depending on the context, it is called a Gaussian pure state \cite{SSM}, a coherent state associated with a generalized Gaussian function \cite[\S1.1.2]{CR}, or a squeezed state \cite[\S2.1]{AAG}. Notably, also the Hagedorn wavepackets coexist in the literature as generalized squeezed states \cite{C} or generalized coherent states \cite[\S3.4]{CR}. The definition of the generalized squeezed or coherent states is less elementary than Hagedorn's ladder approach and uses phase space translations and squeezing operators. We will address it in Section~\S\ref{sec:genstat}. However, both construction tools, the ladders and the translated squeezing operators, balance the position and the momentum parameters. We therefore ask for transforms, which treat position and momentum or time and frequency variables simultaneously. That is, we aim at the Wigner transform and the Fourier-Bros-Iagolnitzer (FBI) transform of the Hagedorn wavepackets and derive explicit formulas built of a Gaussian function and Laguerre polynomials.

The $\eps$-scaled Wigner transform of two Schwartz functions $\phi,\psi:\R^d\to\C$ is the Fourier transform of their correlation function, 
$$
\W^\eps(\phi,\psi)(x,\xi) = 
(2\pi\eps)^{-d} \int_{\R^d} \overline\phi(x+\tfrac{y}{2})\psi(x-\tfrac{y}{2}) \e^{iy^T \xi/\eps} dy,\qquad
(x,\xi)\in\R^{2d},
$$
see for example \cite[Chapter~1.8]{Fo} for a discussion of basic properties or \cite{Br} for the interpretation as a musical score. The Wigner transform maps Schwartz functions to Schwartz functions on phase space, while respecting orthogonality in the sense that
\begin{equation}
\label{eq:orth}
\langle \W^\eps(\phi_1,\psi_1),\W^\eps(\phi_2,\psi_2)\rangle_{L^2(\R^{2d})} = \mbox{$(2\pi\eps)^{-d}$}\;\overline{\langle\phi_1,\phi_2\rangle}_{L^2(\R^d)} \langle\psi_1,\psi_2\rangle_{L^2(\R^d)}
\end{equation}
for all Schwartz functions $\phi_1,\phi_2,\psi_1,\psi_2:\R^d\to\C$. On the diagonal, the Wigner function $\W^\eps(\phi)=\W^\eps(\phi,\phi)$ is real-valued with the position and momentum density as marginals, 
$$
|\phi(x)|^2 = \int_{\R^d} \W^\eps(\phi)(x,\xi) d\xi,\qquad
|\F^\eps\phi(\xi)|^2 = \int_{\R^d} \W^\eps(\phi)(x,\xi) dx.
$$ 
However, $\W^\eps(\phi)$ is not a probability density on phase space, since it might attain negative values. For example, odd functions $\phi$ satisfy
$\W^\eps(\phi)(0,0) = -(2\pi\eps)^{-d}\|\phi\|^2$, and also the Wigner transforms of the Hagedorn wavepackets, except for the initial Gaussian $\phi^\eps_0$, attain negative values. This lack of positivity can be cured by the convolution with the properly $\eps$-scaled Gaussian phase space function $G^\eps(z)=(\pi\eps)^{-d} \exp(-\tfrac{1}{\eps}|z|^2)$, $z\in\R^{2d}$. The resulting positive transform
$$
\H^\eps(\phi) = G^\eps * \W^\eps(\phi)
$$
is the so-called Husimi transform $\H^\eps(\phi):\R^{2d}\to[0,\infty[$. The Husimi transform can also be deduced from the Fourier-Bros-Iagolnitzer (FBI) transform, which is defined as the inner product with the Gaussian wavepacket
$$
g^\eps_{x,\xi}(y) = (\pi\eps)^{-d/4} \exp\!\left(-\tfrac{1}{2\eps}|y-x|^2+\tfrac{i}{\eps}\xi^T(y-x)\right),\qquad y\in\R^d,
$$
centered in the phase space point $(x,\xi)\in\R^{2d}$. That is, 
$$
\T^\eps(\phi)(x,\xi) = (2\pi\eps)^{-d/2} \mbox{$\langle g^\eps_{x,\xi}, \phi\rangle$},\qquad (x,\xi)\in\R^{2d},
$$
see for example \cite[Chapter~3.3]{Fo}. Then, the Husimi transform appears as the modulus squared of the FBI transform, that is, $\H^\eps(\phi)(x,\xi)=|\T^\eps(\phi)(x,\xi)|^2$ for all $(x,\xi)\in\R^{2d}$, 
which immediately reveals positivity. 

Our study of the Hagedorn wavepackets proceeds as follows. In Section \S\ref{sec:herm} we briefly summarize classical results on Hermite functions and recall a polynomial sum rule for computing the Wigner and FBI transform of the Hermite functions. In Section \S\ref{sec:hag} we first review the construction process of the Hagedorn wavepackets. Then, we prove
$$
p^\eps_k(x+z) = \sum_{\nu\le k} \binom{k}{\nu} \left(\tfrac{2}{\sqrt\eps}Q^{-1}z\right)^{k-\nu} p^\eps_\nu(x),\qquad x,z\in\C^d,
$$
for the polynomial factor $p^\eps_k$ of the $k$th Hagedorn wavepacket $\phi^\eps_k$. This new sum rule allows us to extend the classic univariate Hermite--Laguerre connection \cite[\S1.9]{Fo}, \cite[\S12.1]{VK}, \cite[\S23]{W} to multiple dimensions. 
Moreover, we establish the new Rodriguez--type formula 
$$
p^\eps_k(x) = |\phi^\eps_0(x)|^{-2} (-\sqrt\eps Q^{*}\nabla_x)^k |\phi^\eps_0(x)|^2,\qquad x\in\R^d.
$$
Section \S\ref{sec:hagph} derives explicit formulas for the Wigner and the FBI transform of the Hagedorn wavepackets via sum rules and recasts the known three-term recurrence relation of the Hagedorn wavepackets on the Wigner level. In particular, we set 
$$
z=-i(P^T(x-q)-Q^T(\xi-p))
$$
for $(x,\xi)\in\R^{2d}$, that is, $z=\Re(z)+i\Im(z)$ with 
$$
\begin{pmatrix}\Re(z)\\ \Im(z)\end{pmatrix} = F^{-1}\begin{pmatrix}x-q\\\xi-p\end{pmatrix},
$$
and obtain by Theorem~\ref{theo:wigner} 
$$
\W^\eps(\phi^\eps_k,\phi^\eps_k)(x,\xi) = \frac{(-1)^{|k|}}{(\pi\eps)^d} \e^{-\frac1\eps|z|^2} \prod_{j=1}^d L_{k_j}^{(0)}\left(\tfrac2\eps |z_j|^2\right),
$$
where $L_{k_j}^{(0)}$ denotes the $k_j$th Laguerre polynomial. In Section \S\ref{sec:genstat}, we discuss the 
relation of Hagedorn wavepackets and generalized coherent states and provide an alternative derivation of the  Wigner transform via the metaplectic transformation associated with the symplectic matrix $F$. In the appendix, Section \S\ref{app:pol} presents another characterization of the Hagedorn wavepackets based on the polar decomposition of one of the width matrices, while \S\ref{app:we} reformulates the ladder operators in Weyl quantization.

\section{Hermite functions}
\label{sec:herm}

Hermite functions are Hermite polynomials times a Gaussian. They can be generated from the Gaussian 
\begin{equation}
\label{eq:herm}
\phi_0(x) = \pi^{-1/4} \exp(-\tfrac{1}{2}x^2),\qquad x\in\R,
\end{equation}
using the ladder operators $a^\dagger=\tfrac{1}{\sqrt2}(x-\nabla_x)$,
$$
\phi_{k+1} = \tfrac{1}{\sqrt{k+1}}a^\dagger \phi_k,\qquad k\in\N.
$$
The formal adjoint $a=\tfrac{1}{\sqrt2}(x + \nabla_x)$ of the ladder operator $a^\dagger$ allows to descend within the Hermite functions,
$$
\phi_{k} = \tfrac{1}{\sqrt{k+1}}a\phi_{k+1},\qquad k\in\N.
$$

\begin{remark}
The Hermite functions are eigenfunctions of the Fourier transform and the harmonic oscillator $\tfrac12(aa^\dagger+a^\dagger a)=\tfrac12(-\Delta_x+x^2)$, 
$$
\F^1\phi_k = (-i)^{k} \phi_k,\qquad \tfrac12(aa^\dagger+a^\dagger a)\phi_k = (k+\tfrac12)\phi_k
$$
for all $k\in\N$. Moreover, the family $\{\phi_k\mid k\in\N\}$ forms an orthonormal basis of the Hilbert space $L^2(\R)$, see for example \cite[\S7.8]{Tl}. They enjoy the following approximation property: Let $K\in\N$, $s\le K$ and $f:\R^d\to\C$ a Schwartz function. Then, 
$$
\big\|f-\sum_{k<K}\langle \phi_k,f\rangle\phi_k\big\|_{L^2(\R)} \le \left(K(K-1)\cdots(K-s+1)\right)^{-1/2}\left\|a^s f\right\|_{L^2(\R)},
$$
see \cite[Theorem 1.2]{Lu}. 
\end{remark}

\subsection{Hermite polynomials}

Alternatively, the Hermite functions can be written as 
$$
\phi_k(x) = \frac{1}{\sqrt{2^k k!}} h_k(x) \phi_0(x),\qquad x\in\R,
$$
with 
\begin{eqnarray}
h_k(x) 
&=& \exp(x^2)(-\tfrac{d}{dx})^k \exp(-x^2)\label{eq:rod}\\
&=& \sum_{j=0}^{\lfloor k/2\rfloor} \frac{k!}{j!(k-2j)!} (-1)^j (2x)^{k-2j}\label{eq:3r}
\end{eqnarray} 
the $k$th Hermite polynomial. Formula (\ref{eq:rod}) is called Rodriguez formula. Starting from $h_0=1$, the Hermite polynomials can also be generated by repeated application of the ladder operator $b^\dagger=2x-\nabla_x$, 
$$
h_{k+1} = b^\dagger h_k,\qquad k\in\N,
$$
or from the three-term recurrence relation
$$
h_{k+1}(x) = 2xh_k(x)-2kh_{k-1}(x),\qquad k\ge1.
$$
The orthonormality of the Hermite functions implies for the Hermite polynomials
\begin{equation}
\label{eq:ortho}
\int_\R h_k(x) h_{l}(x) \e^{-x^2} dx = \sqrt\pi 2^k k! \delta_{k,l},\qquad k,l\in\N. 
\end{equation}

\subsection{Integral formulas}

The Hermite polynomials satisfy several beautiful integral formulas. Those, which we employ for the phase space transformation of the Hermite functions,  
can be deduced from the following sum rule
\begin{equation}
\label{eq:srh}
h_k(x+z) = \sum_{j=0}^k \binom{k}{j} (2z)^{k-j} h_j(x),\qquad x,z\in\C,
\end{equation}
which is due to \cite{Fe}. We refer to Proposition~\ref{prop:lag} later on for a proof.

\begin{proposition}[Laguerre connection] 
\label{prop:hl}
Let $k\le l$ and $h_k$ and $h_l$ be the $k$th and $l$th Hermite polynomial. Then, for $z_1,z_2\in\C$, 
\begin{equation}
\label{eq:hl}
\int_{\R}  h_k(x+z_1) h_l(x+z_2) \e^{-x^2} dx = \sqrt\pi 2^l k! z_2^{l-k} L^{(l-k)}_k(-2z_1z_2),
\end{equation}
where 
\begin{equation}
\label{eq:lag}
L_k^{(\gamma)} (x) = \sum_{j=0}^k (-1)^j \binom{k+\gamma}{k-j} \frac{x^j}{j!},\qquad k\in\N 
\end{equation}
are the Laguerre polynomials associated with $\gamma\in\R$. In particular,
\begin{equation}
\label{eq:hm}
\int_{\R}  h_l(x+z_2) \e^{-x^2} dx = \sqrt\pi 2^l z_2^{l}.
\end{equation}
\end{proposition}

\begin{proof}
We write
\begin{eqnarray*}
\lefteqn{\int_{\R}  h_k(x+z_1) h_l(x+z_2) \e^{-x^2} dx}\\
&=& \sum_{j=0}^k \sum_{j'=0}^l \binom{k}{j}\binom{l}{j'} (2z_1)^{k-j} (2z_2)^{l-j'} \int_\R h_j(x) h_{j'}(x) \e^{-x^2} dx.
\end{eqnarray*}
From the orthogonality condition (\ref{eq:ortho}) we then deduce
\begin{eqnarray*}
\lefteqn{\int_{\R}  h_k(x+z_1) h_l(x+z_2) \e^{-x^2} dx = \sum_{j=0}^k \binom{k}{j}\binom{l}{j} (2z_1)^{k-j} (2z_2)^{l-j} \sqrt\pi 2^j j!}\\
&=& \sqrt\pi 2^l k! z_2^{l-k} \sum_{j=0}^k \frac{l!}{j!(l-j)!} \frac{(2z_1z_2)^{k-j}}{(k-j)!}
= \sqrt\pi 2^l k! z_2^{l-k} L^{(l-k)}_k(-2z_1z_2).
\end{eqnarray*}
\end{proof}

\subsection{Phase space transforms}

The Hermite-Laguerre connection of Proposition~\ref{prop:hl} translates to the Wigner transform of Hermite functions. For alternative proofs, see \cite[Chapter~1.9]{Fo} or \cite[Chapter~1.3]{Tn}.

\begin{corollary}[Wigner transform]\label{cor:wigner}
If $\phi_k,\phi_l$ are the $k$th and the $l$th Hermite function, then the Wigner function is 
$$
\W^1(\phi_k,\phi_l)(x,\xi)=\left\{
\begin{array}{ll}
\frac{(-1)^{k}}{\pi} \sqrt{2^{l-k}} \sqrt{\frac{k!}{l!}} \overline{z}^{l-k} \e^{-\left|z\right|^2} L^{(l-k)}_{k}(2\left|z\right|^2), & k\le l,\\*[1ex]
\frac{(-1)^{l}}{\pi} \sqrt{2^{k-l}} \sqrt{\frac{l!}{k!}} z^{k-l} \e^{-\left|z\right|^2} L^{(k-l)}_{l}(2\left|z\right|^2), & l\le k,
\end{array}\right.
$$
with $z = x + i\xi$ for $x,\xi\in\R$. In particular, 
$$
\W^1(\phi_k)(x,\xi) = \frac{(-1)^{k}}{\pi} \e^{-\left|z\right|^2} L^{(0)}_{k}(2\left|z\right|^2).
$$
\end{corollary}

\begin{proof} 
We compute
\begin{eqnarray*}
\W^1(\phi_k,\phi_l)(x,\xi) 
&=& \frac{1}{2\pi^{3/2}} \frac{1}{\sqrt{2^{k+l}k! l!}} 
\int_{\R} h_k(x+\tfrac{y}{2}) h_l(x-\tfrac{y}{2}) \e^{-(x^2+(y/2)^2)} \e^{iy\xi} dy\\
&=&
\frac{(-1)^l}{2\pi^{3/2}} \frac{\e^{-|z|^2}}{\sqrt{2^{k+l}k! l!}} 
\int_{\R} h_k(\tfrac{y}{2}+x) h_l(\tfrac{y}{2}-x) \e^{-(y/2-i\xi)^2} dy,
\end{eqnarray*}
where we have used $h_l(-x)=(-1)^l h_l(x)$. Changing the variable as $y/2-i\xi=\eta$, we perform a contour integration in the complex plane. Analyticity and exponential decay of the integrand then provide
$$ 
\W^1(\phi_k,\phi_l)(x,\xi)  = 
\frac{(-1)^l}{\pi^{3/2}} \frac{\e^{-|z|^2}}{\sqrt{2^{k+l}k! l!}} 
\int_{\R} h_k(\eta+z) h_l(\eta-\overline{z}) \e^{-\eta^2} d\eta.
$$
The integral formula (\ref{eq:hl}) concludes the proof.  
\end{proof}

The integral formula (\ref{eq:hm}) of Proposition~\ref{prop:hl} provides the FBI and the Husimi transform of Hermite functions, see \cite[\S2]{Fl}. 

\begin{corollary}[FBI transform] Let $\phi_k$ be the $k$th Hermite function. Then, the FBI transform is
$$
\T^1(\phi_k)(x,\xi) = \frac{\e^{\tfrac{i}{2}x\xi}}{\sqrt{\pi 2^{k+1}k!}} \overline{z}^k \e^{-\tfrac{1}{4}|z|^2},\qquad x,\xi\in\R, 
$$
with $z=x+i\xi$. Consequently, the Husimi transform is 
$$
\H^1(\phi_k)(x,\xi) = \frac{1}{\pi 2^{k+1}k!} |z|^{2k} \e^{-\tfrac{1}{2}|z|^2}.
$$ 
\end{corollary}

\begin{proof} 
We apply the integral formula (\ref{eq:hm}) to obtain 
\begin{eqnarray*}
\T^1(\phi_k)(x,\xi) &=& 
\frac{1}{\pi\sqrt{2^{k+1}k!}} \int_{\R} h_k(y) \e^{-\tfrac12 y^2} \e^{-\tfrac12 (y-x)^2} \e^{-i\xi\cdot(y-x)} dy\\
&=&
\frac{\e^{\tfrac{i}{2}x\xi}\e^{-\tfrac{1}{4}(x^2+\xi^2)}}{\pi\sqrt{2^{k+1}k!}} \int_{\R} h_k(y) \e^{-(y-\tfrac12(x-i\xi))^2} dy\\
&=&
\frac{\e^{\tfrac{i}{2}x\xi}\e^{-\tfrac{1}{4}(x^2+\xi^2)}}{\sqrt{\pi 2^{k+1}k!}}(x-i\xi)^k.
\end{eqnarray*}
\end{proof}

\begin{remark}
\label{rem:en} 
We note that both the Wigner transform $\W^1(\phi_k)$ and the Husimi transform $\H^1(\phi_k)$ of the $k$th Hermite function $\phi_k$ are radially symmetric, that is, they are functions of the energy variable $|z|^2=x^2+\xi^2$, see also \cite[\S2.4]{HC}. 
\end{remark}

\section{Hagedorn wavepackets}
\label{sec:hag}

In \cite{H98}, George Hagedorn devised parametrized ladder operators which allow for a beautiful generalization of the Hermite functions also in several space dimensions: Let $\eps>0$ be a scale parameter, $q,p\in\R^d$ and $C=C^T\in\C^{d\times d}$ a complex symmetric matrix with positive definite imaginary part $\Im(C)>0$. Here and in the following, we denote by
$$
\Re(A) = (\Re(a_{kl}))\in\R^{d\times d},\quad \Im (A) = (\Im(a_{kl}))\in\R^{d\times d}
$$
for $A=(a_{kl})\in\C^{d\times d}$.

\begin{remark}\label{rem:sympl} 
Any complex symmetric matrix $C\in\C^{d\times d}$ with $\Im(C)>0$ can be written as $C=PQ^{-1}$, where $P,Q\in\C^{d\times d}$ are invertible matrices satisfying
\begin{equation}\label{eq:QP}
Q^T P - P^T Q= 0,\qquad Q^*P-P^*Q = 2i \Id.
\end{equation} 
This condition is equivalent to 
$$
F=\begin{pmatrix} 
\Re(Q) & \Im(Q)\\ \Re(P) & \Im(P)
\end{pmatrix}
\,\mbox{is symplectic:}\quad F^TJF=J,\, J=\begin{pmatrix}0 & -\Id\\ \Id & 0\end{pmatrix}.
$$
Conversely, any pair of matrices $P,Q\in\C^{d\times d}$ satisfying (\ref{eq:QP}) is invertible and defines via $C=PQ^{-1}$ a complex symmetric matrix with 
$$
\Im(C)=(QQ^*)^{-1}>0,
$$ 
see \cite[\S 5, Lemma~1.1]{Lu}. The applications of the Hagedorn wavepackets especially to quantum dynamics emphasize the importance of allowing both matrices $Q$ and $P$ to be complex. 
We note, that Hagedorn uses $A$ and $iB$ for $Q$ and $P$ in his work, while we adopt the more recent notation of \cite{FGL,Lu}.
\end{remark}

Let $Q,P\in\C^{d\times d}$ be matrices satisfying (\ref{eq:QP}). Then, the Hermite ladder operators $a^\dagger=\tfrac{1}{\sqrt2}(x-\nabla_x)$ and $a=\tfrac{1}{\sqrt2}(x+\nabla_x)$ generalize as
\begin{eqnarray*}
A^\dagger[q,p,Q,P] &=& \tfrac{i}{\sqrt{2\eps}} \left(P^*(x-q) - Q^*(-i\eps\nabla_x - p)\right),\\*[1ex]
A[q,p,Q,P] &=& -\tfrac{i}{\sqrt{2\eps}} \left(P^T(x-q) - Q^T(-i\eps\nabla_x- p)\right).
\end{eqnarray*}
Before employing this ladder, we summarize some useful linear algebra. 

\begin{lemma}
Let $\eps>0$, $q,p\in\R^{d}$, and $Q,P\in\C^{d\times d}$ be matrices satisfying~(\ref{eq:QP}). Then, 
$QQ^*$ and $PP^*$ are real symmetric matrices. Moreover, the components of $A^\dagger[q,p,Q,P]=(A_j^\dagger)_{j=1}^d$ commute, that is, 
$$[A_j^\dagger,A_k^\dagger]=A_j^\dagger A_k^\dagger-A_k^\dagger A_j^\dagger = 0$$ for $j,k=1,\ldots,d$. 
\end{lemma}

\begin{proof}
Since $C=PQ^{-1}$ is complex symmetric with $\Im(C)=(QQ^*)^{-1}>0$, we have $QQ^*\in\R^{d\times d}$ and $QQ^*=(QQ^*)^T$. Moreover,
$$
PP^* = CQ Q^*\overline C=\Re(C)\Im(C)^{-1}\Re(C) + \Im(C)\in\R^{d\times d},\qquad PP^*=(PP^*)^T.
$$
Finally, 
\begin{eqnarray*}
\left[A_j^\dagger,A_k^\dagger\right] &=& 
\tfrac{1}{2\eps}\Big[
\sum_{l=1}^d \overline p_{lj}x_l-\overline q_{lj}(-i\eps\partial_{x_l}),
\sum_{m=1}^d \overline p_{mk}x_m-\overline q_{mk}(-i\eps\partial_{x_m})
\Big]\\
&=&
-\tfrac{1}{2\eps}\sum_{l,m=1}^d \left(
\overline p_{lj}\overline q_{mk}[x_l,-i\eps\partial_{x_m}] + 
\overline q_{lj}\overline p_{mk}[-i\eps\partial_{\xi_l},x_m]\right)\\
&=&
\tfrac{i}{2} \left(-P^*\overline Q+Q^*\overline P\right)_{jk} = 0
\end{eqnarray*}
due to the canonical commutator relation $\tfrac{1}{i\eps}[x_j,-i\eps\partial_{x_k}]=\delta_{jk}$ and the matrix 
condition~(\ref{eq:QP}).
\end{proof} 

The next step, is to generalize the zeroth order Hermite function $\phi_0$ as the complex Gaussian wavepacket
\begin{eqnarray*}
\phi^\eps_0(x) &=& \phi^\eps_0[q,p,Q,P](x)\\*[1ex]
&=& 
(\pi\eps)^{-d/4} \det(Q)^{-1/2} \exp(\tfrac{i}{2\eps}(x-q)^T PQ^{-1}(x-q)+\tfrac{i}{\eps}p^T(x-q))
\end{eqnarray*}
for $x\in\R^d$. We note, that the vectors $q,p\in\R^d$ and the matrices $QQ^*,PP^*\in\R^{d\times d}$ provide the centers and the width of the corresponding 
position and momentum densities,
\begin{eqnarray*}
|\phi^\eps_0(x)|^2 &=& (\pi\eps)^{-d/2} |\det(Q)|^{-1} \exp(-\tfrac{1}{\eps}(x-q)^T (QQ^*)^{-1}(x-q)),\\
|\F^\eps\phi^\eps_0(\xi)|^2 &=& (\pi\eps)^{-d/2} |\det(P)|^{-1} \exp(-\tfrac{1}{\eps}(\xi-p)^T (PP^*)^{-1}(\xi-p))
\end{eqnarray*}
for $x,\xi\in\R^d$. With the ladder operator $A^\dagger=A^\dagger[q,p,Q,P]$, the Hagedorn wavepackets $\phi^\eps_k=\phi^\eps_k[q,p,Q,P]$ are defined recursively via
\begin{equation}
\label{eq:crea}
\phi_{k+e_j}^\eps = \tfrac{1}{\sqrt{k_j+1}} A^\dagger_j \phi^\eps_k,\qquad
k\in\N^d,
\end{equation}
where $e_j\in\N^d$ denotes the $j$th unit vector for $j=1,\ldots,d$. The formal adjoint of $A^\dagger$ allows to descend,
$$
\phi_{k-e_j}^\eps = \tfrac{1}{\sqrt{k_j}}A_j\phi_{k}^\eps.
$$
The harmonic oscillator equation of the Hermite functions generalizes to
$$
\tfrac12 \sum_{j=1}^d (A_j A_j^\dagger + A_j^\dagger A_j)\phi_k^\eps = (|k|+\tfrac12)\phi_k^\eps,\qquad k\in\N^d.
$$
That is, every Hagedorn wavepacket $\phi^\eps_k$, $k\in\N^d$, is an eigenfunction of an harmonic oscillator for the eigenvalue $|k|+\tfrac12$, see \cite[Theorem 3.3]{H98}. 
The balance between position and momentum parameters is nicely observed in the explicit formula for the $\eps$-scaled Fourier transform,
$$
\F^\eps\phi^\eps_k[q,p,Q,P] = (-i)^{|k|} \e^{-ip^T q/\eps} \phi^\eps_k[p,-q,P,-Q],\qquad k\in\N^d, 
$$
see \cite[\S2]{H98} for an effortless ladder proof. 

\begin{remark} Let $\eps>0$, $q_0,p_0\in\R^d$, and $Q_0,P_0\in\C^{d\times d}$ be matrices satisfying (\ref{eq:QP}). If 
$H=-\tfrac{\eps^2}{2}\Delta_x+V(x)$ 
is a Schr\"odinger operator, whose potential function $V$ has suitable smoothness and growth properties, then
$$
\e^{-iHt/\eps}\phi^\eps_k[q_0,p_0,Q_0,P_0] =  \e^{iS_t/\eps}\phi^\eps_k[q_t,p_t,Q_t,P_t] + O(\eps^{1/2})
$$
as $\eps\to0$, where
\begin{eqnarray*}
\dot q_t &=& p_t, \quad \dot p_t = -\nabla V(q_t),\\
\dot Q_t &=& P_t, \quad \dot P_t = -D^2 V(q_t)Q_t,
\end{eqnarray*}
and $S_t = \int_0^t \left(\tfrac12|p_\tau|^2 - V(q_\tau)\right) d\tau$, see \cite[Theorem~3.5]{H98}. This simple semiclassical propagation result generalizes to arbitrary order in $\eps$ \cite[Theorem~3.6]{H85}, exponential accuracy and the Ehrenfest time scale, see \cite{HJ}. The recent monograph \cite{CR} provides more results on semiclassical wavepacket propagation and further pointers to the literature on coherent states.
\end{remark}

Let $\eps>0$, $q,p\in\R^d$ and
$$
T_{q,p} = \exp\!\left(\tfrac{i}{\eps}\left(p^T x - q^T (-i\eps\nabla_x)\right)\right)
$$
the Heisenberg--Weyl translation operator. It performs a phase space translation on Schwartz functions $\psi:\R^d\to\C$ according to
$$
(T_{q,p}\psi)(x) = \e^{\frac{i}{\eps}p^T (x-\frac12 q)} \psi(x-q),\qquad x\in\R^d,
$$
see e.g. \cite[Definition~124]{Go}. The Wigner and FBI transform as well as the Hagedorn ladders and wavepackets behave conveniently under the action of the translation operator, which will be useful later on.

\begin{lemma}[Phase space translations]
\label{lem:mt}
Let $\eps>0$, $q,p\in\R^d$, and $\phi,\psi:\R^d\to\C$ be Schwartz functions. Then, 
\begin{eqnarray*}
\W^\eps(T_{q,p}\phi,T_{q,p}\psi)(x,\xi) &=& \W^\eps(\phi,\psi)(x-q,\xi-p),\\
\T^\eps(T_{q,p}\psi)(x,\xi) &=& \e^{\frac{i}{\eps}p^T(x-\frac12q)} \T^\eps(\psi)(x-q,\xi-p).
\end{eqnarray*}
Moreover, for matrices $Q,P\in\C^{d\times d}$ satisfying (\ref{eq:QP}) we have 
\begin{eqnarray}\label{eq:transhag}
T_{q,p}^{-1} \;A^\dagger[q,p,Q,P]\; T_{q,p} &=& A^\dagger[0,0,Q,P],\\
T_{q,p}^{-1} \;A[q,p,Q,P] \; T_{q,p} &=& A[0,0,Q,P].\nonumber 
\end{eqnarray}
In particular,
$T_{q,p}\,\phi^\eps_k[0,0,Q,P]= \e^{\tfrac{i}{2\eps}p^Tq}\phi^\eps_k[q,p,Q,P]$ for all $k\in\N^d$.
\end{lemma}

\begin{proof}
We check
\begin{eqnarray*}
\W^\eps(T_{q,p}\phi,T_{q,p}\psi)(x,\xi) &=& 
(2\pi\eps)^{-d} \int_{\R^d} \overline\phi(x+\tfrac{y}{2}-q)\psi(x-\tfrac{y}{2}-q)\e^{iy^T(\xi-p)/\eps} dy\\
&=&
\W^\eps(\phi,\psi)(x-q,\xi-p)
\end{eqnarray*}
and 
\begin{eqnarray*}
\lefteqn{\T^\eps(T_{q,p}\psi)(x,\xi)}\\ 
&=&
(2\pi\eps)^{-d/2} (\pi\eps)^{-d/4} 
\int_{\R^d}\psi(y-q)\e^{\frac{i}{\eps}p^T(y-\frac12q)-\tfrac{1}{2\eps}|y-x|^2-\tfrac{i}{\eps}\xi^T(y-x)} dy\\
&=&
(2\pi\eps)^{-d/2} (\pi\eps)^{-d/4} 
\int_{\R^d}\psi(y)\e^{\frac{i}{\eps}p^T(y+\frac12q)-\tfrac{1}{2\eps}|y-(x-q)|^2-\tfrac{i}{\eps}\xi^T(y-(x-q))} dy\\
&=& 
(2\pi\eps)^{-d/2} (\pi\eps)^{-d/4} \e^{\frac{i}{\eps}p^T(x-\frac12q)}
\int_{\R^d}\psi(y)\e^{-\tfrac{1}{2\eps}|y-(x-q)|^2-\tfrac{i}{\eps}(\xi-p)^T(y-(x-q))} dy\\
&=&
\e^{\frac{i}{\eps}p^T(x-\frac12q)}\T^\eps(\psi)(x-q,\xi-p).
\end{eqnarray*}
We compute 
$$
T_{q,p} \,x\, T_{q,p}^{-1} = x-q,\qquad 
T_{q,p}(-i\eps\nabla_x)T_{q,p}^{-1} = -i\eps\nabla_x-p,
$$ 
and obtain (\ref{eq:transhag}). Since $T_{q,p}\phi^\eps_0[0,0,Q,P]=\e^{\tfrac{i}{2\eps}p^Tq}\phi^\eps_0[q,p,Q,P]$,
we have for all $k\in\N^d$, 
\begin{eqnarray*}
\lefteqn{T_{q,p}\phi^\eps_k[0,0,Q,P]= \tfrac{1}{\sqrt{k!}} \,T_{q,p}\, (A^\dagger[0,0,Q,P])^k \phi^\eps_0[0,0,Q,P]}\\
&=& 
\tfrac{1}{\sqrt{k!}} (A^\dagger[q,p,Q,P])^k \,T_{q,p}\,\phi^\eps_0[0,0,Q,P]\\
&=&
\tfrac{1}{\sqrt{k!}}\, \e^{\tfrac{i}{2\eps}p^Tq} (A^\dagger[q,p,Q,P])^k \phi^\eps_0[q,p,Q,P] =
\e^{\tfrac{i}{2\eps}p^Tq}\, \phi^\eps_k[q,p,Q,P].
\end{eqnarray*}
\end{proof}

\subsection{Hermite polynomials}
The Hagedorn wavepackets are known to satisfy the three-term recurrence relation
\begin{eqnarray}
\label{eq:3t}
\lefteqn{\left(\sqrt{k_j+1}\phi^\eps_{k+e_j}(x)\right)_{j=1}^d =}\\
&& 
\sqrt{\tfrac{2}{\eps}}Q^{-1}(x-q)\phi^\eps_k(x) - Q^{-1}\overline Q\left(\sqrt{k_j}\phi^\eps_{k-e_j}(x)\right)_{j=1}^d,
\nonumber
\end{eqnarray}
see \cite[Chapter~V.2]{Lu}. Therefore, $\phi^\eps_k=\phi^\eps_k[q,p,Q,P]$ is the product of a polynomial of degree $|k|$ with the complex Gaussian $\phi^\eps_0$. However, we can deduce some new information. 
 
\begin{proposition}[Polynomial ladder]
\label{prop:pol}
Let $\eps>0$, $q,p\in\R^d$, and $Q,P\in\C^{d\times d}$ satisfy~(\ref{eq:QP}). The $k$th Hagedorn wavepacket $\phi^\eps_k=\phi^\eps_k[q,p,Q,P]$, $k\in\N^d$, can be written as
\begin{equation}
\label{eq:pol}
\phi^\eps_k(x) = \frac{1}{\sqrt{2^{|k| }k!}}\,p^\eps_k(x) \phi^\eps_0(x),\qquad x\in\R^d,
\end{equation}
where $p^\eps_k$ is a multivariate polynomial of degree $|k|$ generated by the recursion
$p^\eps_0=1$, $p^\eps_{k+e_j}=B^\dagger_j p^\eps_k$ for $j=1,\ldots,d$, 
with 
$$
B^\dagger = \tfrac{2}{\sqrt\eps} Q^{-1}\op_\eps(x-q) - \tfrac{i}{\sqrt\eps}Q^*(-i\eps\nabla_x).
$$ 
The components of 
$B^\dagger=(B_j^\dagger)_{j=1}^d$ commute, that is, $[B_j^\dagger,B_{j'}^\dagger]=0$ for $j,j'=1,\ldots,d$.  
Moreover, 
\begin{equation}
\label{eq:sym}
p^\eps_k(-x)=(-1)^{|k|} p^\eps_k(x+2q),\qquad x\in\R^{d}.
\end{equation} 
If $Q=Q^T\in\R^{d\times d}$, then 
\begin{equation}
\label{eq:prod}
p_k^\eps(x) = \prod_{j=1}^d h_{k_j}\!\left(\tfrac{1}{\sqrt\eps}(Q^{-1}(x-q))_j\right),\qquad x\in\R^d.
\end{equation}
\end{proposition}

\begin{proof}We denote $A^\dagger=A^\dagger[q,p,Q,P]$.  We first consider the special case $Q=Q^T\in\R^{d\times d}$ and compute
\begin{eqnarray*}
\lefteqn{Q^*(-i\eps\nabla_x-p)\prod_{j=1}^d h_{k_j}(y_j)\;\phi_0^\eps(x)}\\
&=&
-i\sqrt\eps Q^*Q^{-T} 
\begin{pmatrix} 
h_{k_1}'(y_1)\prod_{j\neq1} h_{k_j}(y_j)\\ \vdots\\ 
h_{k_d}'(y_d)\prod_{j\neq d} h_{k_j}(y_j)
\end{pmatrix}\phi^\eps_0(x)\\
&&
+Q^*C(x-q)\prod_{j=1}^d h_{k_j}(y_j)\phi^\eps_0(x)
\end{eqnarray*}
with $C=PQ^{-1}$. Since $P^*Q-Q^*C Q=-2i\,\Id$, we obtain
\begin{eqnarray*}
A^\dagger \prod_{j=1}^d h_{k_j}(y_j)\;\phi_0^\eps(x)&=&
\tfrac{1}{\sqrt 2} \begin{pmatrix} 
\left(2y_1h_{k_1}(y_1)-h_{k_1}'(y_1)\right)\prod_{j\neq1} h_{k_j}(y_j)\\ \vdots\\ 
\left(2y_dh_{k_d}(y_d)-h_{k_d}'(y_d)\right)\prod_{j\neq d} h_{k_j}(y_j)
\end{pmatrix}\phi^\eps_0(x)\\
&=&
\tfrac{1}{\sqrt 2} \begin{pmatrix} 
h_{k_1+1}(y_1)\prod_{j\neq1} h_{k_j}(y_j)\\ \vdots\\ 
h_{k_d+1}(y_d)\prod_{j\neq d} h_{k_j}(y_j)
\end{pmatrix}\phi^\eps_0(x).
\end{eqnarray*}
Assuming, that the claimed identity (\ref{eq:prod}) holds for $k\in\N^d$, we derive 
$$
\phi_{k+e_j}^\eps(x) = \tfrac{1}{\sqrt{k_j+1}}A^\dagger_j\phi_k^\eps(x) = 
\tfrac{1}{\sqrt{2^{|k|+1}(k+e_j)!}} h_{k_j+1}(y_j)\prod_{l\neq j} h_{k_l}(y_l)\phi^\eps_0(x).
$$
For general $Q\in\C^{d\times d}$, we observe that
$$
Q^*(-i\eps\nabla_x-p)\,p^\eps_k(x)\phi^\eps_0(x) = \left(Q^*(-i\eps\nabla_x)p^\eps_k\right)(x)\phi^\eps_0(x)+Q^*C(x-q)p^\eps_k(x)\phi^\eps_0(x)
$$
and
\begin{eqnarray*}
A^\dagger p^\eps_k(x)\phi^\eps_0(x) 
&=& 
\tfrac{i}{\sqrt{2\eps}}\left((P^*-Q^*C)(x-q)p^\eps_k(x)-Q^*(-i\eps\nabla_x)p^\eps_k(x)\right)\phi^\eps_0(x)\\
&=&
\tfrac{1}{\sqrt2}\left(\tfrac{2}{\sqrt\eps} Q^{-1}(x-q)p^\eps_k(x)-\tfrac{i}{\sqrt{\eps}}Q^*(-i\eps\nabla_x)p^\eps_k(x)\right)\phi^\eps_0(x). 
\end{eqnarray*}
Assuming that equation (\ref{eq:pol}) holds for $k\in\N^d$, we conclude
$$
\phi^\eps_{k+e_j} = \tfrac{1}{\sqrt{k_j+1}}A^\dagger_j\phi^\eps_k = 
\tfrac{1}{\sqrt{2^{|k|+1}(k+e_j)!}} \left(B^\dagger_j p^\eps_k\right) \phi^\eps_0.
$$
Moreover, 
\begin{eqnarray*}
\lefteqn{[B_j^\dagger,B_{j'}^\dagger]}\\
&=& -2i \sum_{l,m=1}^d \left(Q^{-1}_{jl}Q^*_{j'm} [x_l,-i\eps\partial_{x_m}] -Q^*_{jl}Q^{-1}_{{j'}m}[-i\eps\partial_{x_l},x_m]\right)\\
&=& 2\eps (Q^{-1}\overline Q-Q^*Q^{-T})_{jj'} = 0,
\end{eqnarray*}
where we have used that $[x_l,-i\eps\partial_{x_m}]=i\eps\delta_{lm}$ and $QQ^*=(QQ^*)^T=\overline Q Q^T$.
For proving the symmetry relation, we argue once more inductively. We have
\begin{eqnarray*}
p^\eps_{k+e_j}(-x) 
&=& 
\tfrac{2}{\sqrt\eps} \left(Q^{-1}(-x-q)\right)_j p^\eps_k(-x) - \sqrt\eps \left((Q^*\nabla_x)_j p^\eps_k\right)(-x)\\
&=&
\tfrac{2}{\sqrt\eps} \left(Q^{-1}(-x-q)\right)_j (-1)^{|k|}p^\eps_k(x+2q)\\ 
&&
+ \sqrt\eps (Q^*\nabla_x)_j (-1)^{|k|} p^\eps_k(x+2q)\\
&=&
(-1)^{|k|+1} p^\eps_{k+e_j}(x+2q).
\end{eqnarray*}
\end{proof}

\begin{remark} In the univariate case, the previous Proposition~\ref{prop:pol} and Proposition~\ref{prop:kron} in Appendix~\ref{app:pol} provide the orginal definition of $\phi_k^\eps[q,p,Q,P]$ given in \cite[\S1]{H81},
$$ 
\phi_k^\eps[q,p,Q,P](x) = \frac{(\overline Q/|Q|)^k}{\sqrt{2^{k}k!}}\,h_{k}\!\left(\frac{x-q}{\sqrt\eps|Q|}\right)\;\phi_0^\eps[q,p,Q,P](x),
\qquad x\in\R.
$$
For $Q=Q^T\in\R^{d\times d}$, the description of Hagedorn wavepackets as rotated, scaled versions of harmonic oscillator eigenfunctions is mentioned in \cite[Remark~4 of \S1]{H85}.  
For general $Q\in\C^{d\times d}$, there is a more complicated relation, see Proposition~\ref{prop:kron}. 
\end{remark}

\subsection{Laguerre polynomials}
The polynomial ladder of Proposition~\ref{prop:pol} allows to generalize the Hermite polynomial's sum rule (\ref{eq:srh}) and the integral connection to the Laguerre polynomials. A similar argumentation for generalized univariate Hermite polynomials can be found in \cite[Proposition~2]{DM}. 

\begin{proposition}[Sum rule \& Laguerre connection]
\label{prop:lag}
Let $\eps>0$, $q,p\in\R^d$, and $Q,P\in\C^{d\times d}$ satisfy (\ref{eq:QP}). Let $k,l\in\N^d$ and $p^\eps_k$, $p^\eps_l$ be the $k$th and $l$th polynomials defined in Proposition~\ref{prop:pol}. Then, for $x,y,z\in\C^d$, 
\begin{equation}
\label{eq:sr}
p^\eps_k(x+z) = \sum_{\nu\le k} \binom{k}{\nu} \left(\tfrac{2}{\sqrt\eps}Q^{-1}z\right)^{k-\nu} p^\eps_\nu(x)
\end{equation}
and
$$
\int_{\R^d} \overline{p^\eps_k}(x+y) p^\eps_l(x+z) |\phi^\eps_0(x)|^2 dx =
\prod_{j=1}^d {\mathcal L}_{k_j,l_j}\!\left(\tfrac{1}{\sqrt\eps}(\overline{Q^{-1}}y)_j,\tfrac{1}{\sqrt\eps}(Q^{-1}z)_j\right)
$$
with
$$
{\mathcal L}_{m,n}(\eta,\zeta) = 
\left\{ \begin{array}{ll}
2^n m!\, \zeta^{n-m} L_m^{(n-m)}(-2\eta\zeta), & m\le n,\\
2^m n!\, \zeta^{m-n} L_n^{(m-n)}(-2\eta\zeta), & n\le m.
\end{array}
\right.
$$
\end{proposition}

\begin{proof}
For $z\in\C^d$ we use the translation $(\tau_z f)(x)=f(x+z)$ and observe 
$$
\tau_z\circ B^\dagger = (B^\dagger+\tfrac{2}{\sqrt\eps}Q^{-1}z)\circ\tau_z.
$$ 
Inductively, we obtain 
$$
\tau_z\circ \left(B^\dagger\right)^k = \left(B^\dagger+\tfrac{2}{\sqrt\eps}Q^{-1}z\right)^k\circ\tau_z = 
\left(\sum_{\nu\le k} \binom{k}{\nu} \left(\tfrac{2}{\sqrt\eps}Q^{-1}z\right)^{k-\nu} (B^\dagger)^\nu\right)
\circ\tau_z  
$$
for all $k\in\N^d$, which applied to $p^\eps_0=1$ yields the claimed sum rule. Therefore, 
\begin{eqnarray*}
\lefteqn{
\int_{\R^d} \overline{p^\eps_k}(x+y) p^\eps_l(x+z) |\phi^\eps_0(x)|^2 dx}\\
&=& 
\sum_{\nu\le k}\sum_{\nu'\le l} \binom{k}{\nu} \binom{l}{\nu'}
\left(\tfrac{2}{\sqrt\eps} \overline{Q^{-1}}y\right)^{k-\nu} 
\left(\tfrac{2}{\sqrt\eps} Q^{-1}z\right)^{l-\nu'} 2^{|\nu|} \nu! \,\delta_{\nu,\nu'}\\
&=&
\sum_{\nu_1=0}^{\min(k_1,l_1)}\cdots \sum_{\nu_d=0}^{\min(k_d,l_d)}\frac{k!l!}{(k-\nu)!(l-\nu)! \nu!}\, 2^{|\nu|}
\left(\tfrac{2}{\sqrt\eps} \overline{Q^{-1}}y\right)^{k-\nu} 
\left(\tfrac{2}{\sqrt\eps} Q^{-1}z\right)^{l-\nu}, 
\end{eqnarray*}
where we have used the orthogonality relation
$$
\int_{\R^d}\overline{p^\eps_\nu}(x) p^\eps_{\nu'}(x)  |\phi^\eps_0(x)|^2 dx = 2^{|\nu|} \nu! \,\delta_{\nu,\nu'}.
$$
If $k_j\le l_j$, then
\begin{eqnarray*}
\lefteqn{
\sum_{\nu_j=0}^{k_j} \frac{k_j!l_j!}{(k_j-\nu_j)!(l_j-\nu_j)! \nu_j!}\, 2^{\nu_j}
\left(\tfrac{2}{\sqrt\eps} (\overline{Q^{-1}}y)_j\right)^{k_j-\nu_j} 
\left(\tfrac{2}{\sqrt\eps} (Q^{-1}z)_j\right)^{l_j-\nu_j}}\\
&=&
k_j!\, 2^{l_j} \left(\tfrac{1}{\sqrt\eps} (Q^{-1}z)_j\right)^{l_j-k_j} 
L^{(l_j-k_j)}_{k_j}\!\left(-\tfrac{2}{\eps} (\overline{Q^{-1}}y)_j (Q^{-1}z)_j\right)
\end{eqnarray*}
by the monomial representation of the Laguerre polynomials (\ref{eq:lag}). The analogous argument for $k_j\ge l_j$ concludes the proof. 
\end{proof}

\subsection{Rodriguez formula}

The Hagedorn wavepackets' three-term recurrence formula~(\ref{eq:3t}) rewrites on the polynomial level as
\begin{equation}
\label{eq:3tp}
\left(p^\eps_{k+e_j}(x)\right)_{j=1}^d = \tfrac{2}{\sqrt\eps} Q^{-1}(x-q)p^\eps_k(x) - 2 Q^{-1}\overline Q \left(k_j p^\eps_{k-e_j}(x)\right)_{j=1}^d. 
\end{equation}
Together with the following new Rodriguez-type formula (\ref{eq:rodnew}) we obtain another integration formula for the polynomials.

\begin{proposition}[Rodriguez formula \& sum rule]
\label{prop:rod}
Let $\eps>0$, $q,p\in\R^d$, and $Q,P\in\C^{d\times d}$ satisfy~(\ref{eq:QP}). Let $k\in\N^d$, and $p^\eps_k$ be the $k$th polynomial defined in Proposition~\ref{prop:pol}. Then,
\begin{equation}
\label{eq:rodnew}
p^\eps_k(x) = |\phi^\eps_0(x)|^{-2} (-\sqrt\eps Q^*\nabla_x)^k |\phi^\eps_0(x)|^2,\qquad x\in\R^d.
\end{equation}
Let $M=M^T\in\C^{d\times d}$ be such that $\Im(C)+\Re(M)>0$ and $\Id+Q^*MQ$ is invertible. Then, we have for all $z\in\C^d$ 
$$
\int_{\R^d} p^\eps_k(x+z) \e^{-\tfrac1\eps (x-q)^T(\Im(C)+M)(x-q)}dx = 
\sum_{\nu\le k} \binom{k}{\nu} \left(\tfrac{2}{\sqrt\eps}{Q^{-1}}z\right)^{k-\nu} c_\nu,
$$
where for all $\nu\in\N^d$ with $|\nu|$ odd 
$$
c_\nu = 0,\qquad (c_{\nu+e_j})_{j=1}^d = -2(\Id+Q^*MQ)^{-1}Q^*M\overline{Q}(\nu_j c_{\nu-e_j})_{j=1}^d.
$$ 
In particular ($M=0$), 
$$
\int_{\R^d} p^\eps_k(x+z) \e^{-\tfrac1\eps (x-q)^T \Im(C)(x-q)}dx = 
(\pi\eps)^{d/2} \det(Q) \left(\tfrac{2}{\sqrt\eps}Q^{-1}z\right)^{k}.
$$
\end{proposition}

\begin{proof}
We set $q^\eps_k(x)=|\phi^\eps_0(x)|^{-2} (-\sqrt\eps Q^*\nabla_x)^k |\phi^\eps_0(x)|^2$ and verify, that $q^\eps_k$ satisfies the three-term recurrence~(\ref{eq:3tp}). By the Leibniz rule, 
\begin{eqnarray*}
\lefteqn{\left(q^\eps_{k+e_j}(x)\right)_{j=1}^d = 
|\phi^\eps_0(x)|^{-2} \left((-\sqrt\eps Q^*\nabla_x)^{k+e_j}\right)_{j=1}^d |\phi^\eps_0(x)|^2}\\
&=&
|\phi^\eps_0(x)|^{-2} (-\sqrt\eps Q^*\nabla_x)^k \tfrac{2}{\sqrt\eps}Q^{-1}(x-q) |\phi^\eps_0(x)|^2\\
&=&
|\phi^\eps_0(x)|^{-2} \sum_{\nu\le k} \binom{k}{\nu} (-\sqrt\eps Q^*\nabla_x)^{\nu} \tfrac{2}{\sqrt\eps}Q^{-1}(x-q) (-\sqrt\eps Q^*\nabla_x)^{k-\nu}|\phi^\eps_0(x)|^2\\
&=&
|\phi^\eps_0(x)|^{-2} \tfrac{2}{\sqrt\eps}Q^{-1}(x-q) (-\sqrt\eps Q^*\nabla_x)^{k}|\phi^\eps_0(x)|^2\\
&-&
2 Q^{-1}|\phi^\eps_0(x)|^{-2} \sum_{j=1}^d k_j \left((Q^*\nabla_x)_j x\right) (-\sqrt\eps Q^*\nabla_x)^{k-e_j}|\phi^\eps_0(x)|^2\\
&=&
\tfrac{2}{\sqrt\eps}Q^{-1}(x-q) q^\eps_k(x) - 2 Q^{-1}\overline{Q} \left(k_jq^\eps_{k-e_j}(x)\right)_{j=1}^d.
\end{eqnarray*}
For proving the claimed integral formula, we use the sum rule (\ref{eq:sr}) and obtain
$$
\int_{\R^d} p^\eps_k(x+z) \e^{-\tfrac1\eps (x-q)^T(\Im(C)+M)(x-q)}dx
=\sum_{\nu\le k} \binom{k}{\nu} \left(\tfrac{2}{\sqrt\eps}Q^{-1}z\right)^{k-\nu} c_\nu
$$
with 
$$
c_\nu = \int_{\R^d} p^\eps_\nu(x) \e^{-\tfrac1\eps (x-q)^T(\Im(C)+M)(x-q)}dx.
$$
By the symmetry relation (\ref{eq:sym}), $c_\nu=0$ for $|\nu|$ odd. Moreover, by the Rodriguez formula and the three-term recurrence,   
\begin{eqnarray*}
\lefteqn{(c_{\nu+e_j})_{j=1}^d}\\ 
&=& 
\int_{\R^d} \e^{-\tfrac1\eps (x-q)^T M(x-q)}
\left((-\sqrt\eps Q^*\nabla_x)^{\nu+e_j}\right)_{j=1}^d |\phi^\eps_0(x)|^2 dx\\
&=&
-\int_{\R^d} \tfrac{2}{\sqrt\eps} Q^* M(x-q) \e^{-\tfrac1\eps (x-q)^T M(x-q)}
(-\sqrt\eps Q^*\nabla_x)^{\nu} |\phi^\eps_0(x)|^2 dx\\
&=&
-Q^* M Q \int_{\R^d} \tfrac{2}{\sqrt\eps} Q^{-1}(x-q) {p^\eps_\nu}(x) \e^{-\tfrac1\eps (x-q)^T(\Im(C)+M)(x-q)} dx\\
&=&
-Q^* M Q (c^\eps_{\nu+e_j})_{j=1}^d - 
2 Q^* M \overline{Q} (\nu_j c^\eps_{\nu-e_j})_{j=1}^d.
\end{eqnarray*}
\end{proof}

\begin{remark} In the univariate case, the previous integral formula allows for simplifications. We set 
$\tfrac{1}{\alpha_1}=\Im(C)+M$ and $\alpha_2=(\overline{Q})^2M/(1+|Q|^2M)$ to obtain
$$
\int_{\R} p^\eps_k(x+z) \e^{-\tfrac{1}{\alpha_1\eps} (x-q)^2}dx = 
\sum_{j=0}^{\lfloor k/2\rfloor} \binom{k}{2j} \left(\tfrac{2}{\sqrt\eps}Q^{-1}z\right)^{k-2j}c_j
$$
with 
\begin{eqnarray*}
c_j
&=&
-2\alpha_2(2j-1)c_{j-1}=2^j(-\alpha_2)^j (2j-1)\cdot(2j-3)\cdots1\cdot c_0\\
&=&
\sqrt{\pi\eps\alpha_1}(-\alpha_2)^j (2j)!/j.
\end{eqnarray*}
The monomial representation of the Hermite polynomials (\ref{eq:3r}) then implies 
$$
\int_{\R} p^\eps_k(x+z) \e^{-\tfrac{1}{\alpha_1\eps} (x-q)^2}dx = \sqrt{\pi\eps\alpha_1} \alpha_2^{k/2} h_k\!\left(\tfrac{1}{\sqrt{\alpha_2\eps}}Q^{-1}z\right).
$$
For the Hermite polynomials, this formula is due to \cite{Fe}. 
\end{remark}

\section{Hagedorn wavepackets in phase space}
\label{sec:hagph}

Our studies so far have provided two integral formulas for the polynomial part of the Hagedorn wavepackets. Now we apply them for computing the Wigner and FBI transform. 

\subsection{Wigner transform}

Proposition's \ref{prop:lag} integral connection to the Laguerre polynomials allows us to write the Wigner function of the Hagedorn wavepackets in terms of Gaussians and Laguerre polynomials depending on the complex vector
$$
z(x,\xi) = -i\left(P^T(x-q) - Q^T(\xi-p)\right),\qquad x,\xi\in\R^d.
$$

\begin{theorem}[Wigner transform, 1st proof]\label{theo:wigner}
Let $\eps>0$, $q,p\in\R^d$ and $Q,P\in\C^{d\times d}$ satisfy (\ref{eq:QP}). Then the scaled Wigner function of the $k$th and the $l$th Hagedorn wavepacket $\phi_k^\eps=\phi^\eps_k[q,p,Q,P]$ and $\phi_l^\eps=\phi^\eps_l[q,p,Q,P]$, $k,l\in\N^d$, satisfies
\begin{equation}
\label{eq:hagwig}
\W^\eps(\phi_k^\eps,\phi_l^\eps)(x,\xi) = 
(\pi\eps)^{-d} \e^{-\tfrac{1}{\eps}|z|^2} \frac{(-1)^{|l|}}{\sqrt{2^{|k|+|l|}k! l!}} \prod_{j=1}^d  
{\mathcal L}_{k_j,l_j}(\tfrac{1}{\sqrt\eps}z_j)
\end{equation}
with $z= -i\left(P^T(x-q) - Q^T(\xi-p)\right)$ for $(x,\xi)\in\R^{2d}$ and 
$$
{\mathcal L}_{m,n}(\zeta) = \left\{
\begin{array}{ll}
2^n m!\, \zeta^{n-m} L^{(n-m)}_m\!\left(2|\zeta|^2\right), 
& m\le n,\\
2^m n!\, \left(-\overline{\zeta}\right)^{m-n} L^{(m-n)}_n\!\left(2|\zeta|^2\right), & n\le m.
\end{array}\right.
$$
In particular, 
$$
\W^\eps(\phi_k^\eps,\phi_k^\eps)(x,\xi) = \frac{(-1)^{|k|}}{(\pi\eps)^d} \e^{-\tfrac{1}{\eps}|z|^2} \prod_{j=1}^d L^{(0)}_{k_j}\!\left(\tfrac{2}{\eps}|z_j|^2\right),\qquad (x,\xi)\in\R^{2d}.
$$
\end{theorem}

\begin{proof} 
We first study the Wigner transform in the origin. By Proposition \ref{prop:pol},
\begin{eqnarray*}
\lefteqn{\W^\eps(\phi_k^\eps,\phi_l^\eps)(0,0)}\\
&=&
(\pi\eps)^{-d} \frac{1}{\sqrt{2^{|k|+|l|}k! l!}}
\int_{\R^d} \overline{p^\eps_k}(y) p^\eps_l(-y)\overline{\phi^\eps_0}(y)\phi^\eps_0(-y) dy\\
&=&
(\pi\eps)^{-d} \frac{(-1)^{|l|}}{\sqrt{2^{|k|+|l|}k! l!}}
\int_{\R^d} \overline{p^\eps_k}(\eta-q) p^\eps_l(\eta+q)\overline{\phi^\eps_0}(\eta-q)\phi^\eps_0(-\eta+q) d\eta.
\end{eqnarray*}
The idea is to simplify the integrand such that the Laguerre connection of Proposition~\ref{prop:lag} can be applied. We compute
\begin{eqnarray*}
\lefteqn{\overline{\phi^\eps_0}(\eta-q)\phi^\eps_0(-\eta+q)}\\
&=& 
(\pi\eps)^{-d/2} \det(Q)^{-1} 
\e^{-\tfrac{i}{2\eps}(\eta-2q)^T \overline C (\eta-2q) +\tfrac{i}{2\eps}\eta^T C \eta -\tfrac{2i}{\eps}p^T(\eta-q)}\\
&=&
(\pi\eps)^{-d/2} \det(Q)^{-1} 
\e^{-\tfrac1\eps \eta^T\Im(C)\eta +\tfrac2\eps\eta^T\Im(C)q -\tfrac2\eps q^T\Im(C)q+\tfrac{2i}{\eps}(\eta-q)^T(\Re(C)q-p)}.
\end{eqnarray*}
The complex vector
\begin{eqnarray*}
z =-i\left(P^T(-q) - Q^T(-p)\right)
\end{eqnarray*}
satisfies $z=iQ^T(Cq-p)$, since $C$ is complex symmetric. Due to $\Im(C)=(QQ^*)^{-1}$ and $QQ^*=(QQ^*)^T=\overline{Q}Q^T$, we have $\overline{Q}z=-q+i|Q|^2(\Re(C)q-p)$. Moreover, 
$$
|z|^2 = \overline{z}^T z = q^T\Im(C)q + (\Re(C)q-p)^T |Q|^2(\Re(C)q-p).
$$ 
We set $w=-i\Im(\overline{Q}z)=-i|Q^2|(\Re(C)q-p)$ and compute
\begin{eqnarray*}
|\phi^\eps_0(\eta+w)|^2
&=&
(\pi\eps)^{-d/2}\det(Q)^{-1} \e^{-\frac1\eps (\eta+w-q)^T\Im(C)(\eta+w-q)}\\
&=&
(\pi\eps)^{-d/2} \det(Q)^{-1} 
\e^{-\tfrac1\eps q^T\Im(C)q + \tfrac{1}{\eps}(\Re(C)q-p)^T |Q|^2(\Re(C)q-p)}\\
&& 
\e^{-\tfrac{1}{\eps}\eta^T\Im(C)\eta + \tfrac{2}{\eps}\eta^T\Im(C)q+\tfrac{2i}{\eps}(\eta-q)^T (\Re(C)q-p)}.
\end{eqnarray*} 
Therefore,
$$
\overline{\phi^\eps_0}(\eta-q)\phi^\eps_0(-\eta+q) = \e^{-\tfrac1\eps|z|^2} |\phi^\eps_0(\eta+w)|^2
$$
and 
\begin{eqnarray*}
\lefteqn{\W^\eps(\phi_k^\eps,\phi_l^\eps)(0,0)}\\ 
&=&
(\pi\eps)^{-d}  \frac{(-1)^{|l|}\e^{-\tfrac1\eps|z|^2}}{\sqrt{2^{|k|+|l|}k! l!}} 
\int_{\R^d+w} \overline{p^\eps_k}(y-w-q) p^\eps_l(y-w+q)|\phi^\eps_0(y)|^2 dy\\
&=&
(\pi\eps)^{-d}  \frac{(-1)^{|l|}\e^{-\tfrac1\eps|z|^2}}{\sqrt{2^{|k|+|l|}k! l!}} 
\int_{\R^d} \overline{p^\eps_k}(y-w-q) p^\eps_l(y-w+q)|\phi^\eps_0(y)|^2 dy
\end{eqnarray*}
due to analyticity and exponential decay of the integrand. By Proposition~\ref{prop:lag},
\begin{eqnarray*}
\lefteqn{\W^\eps(\phi_k^\eps,\phi_l^\eps)(0,0)}\\ 
&=&
(\pi\eps)^{-d}  \frac{(-1)^{|l|}\e^{-\tfrac1\eps|z|^2}}{\sqrt{2^{|k|+|l|}k! l!}}
\prod_{j=1}^d {\mathcal L}_{k_j,l_j}\!
\left(-\tfrac{1}{\sqrt\eps}(\overline{Q^{-1}}(q+w))_j,\tfrac{1}{\sqrt\eps}(Q^{-1}(q-w))_j\right)\\
&=&
(\pi\eps)^{-d}  \frac{(-1)^{|l|}\e^{-\tfrac1\eps|z|^2}}{\sqrt{2^{|k|+|l|}k! l!}}
\prod_{j=1}^d {\mathcal L}_{k_j,l_j}\!
\left(\tfrac{1}{\sqrt\eps}z_j,-\tfrac{1}{\sqrt\eps}\overline{z_j}\right), 
\end{eqnarray*}
where we have used that $-q-w=\overline{Q}z$ and $q-w=-Q\overline{z}$. For arbitrary $(x,\xi)\in\R^{2d}$, we use Lemma~\ref{lem:mt} and obtain
\begin{eqnarray*}
\W^\eps(\phi_k^\eps,\phi_l^\eps)(x,\xi) 
&=& \W^\eps(T_{x,\xi}^{-1}\phi^\eps_k,T_{x,\xi}^{-1}\phi^\eps_l)(0,0)\\
&=& \W^\eps(\phi^\eps_k[q-x,p-\xi,Q,P],\phi^\eps_l[q-x,p-\xi,Q,P])(0,0).
\end{eqnarray*}
\end{proof}

To generalize the observation, that the Hermite function's Wigner function only depends on the energy variable $|z|^2$, see Remark~\ref{rem:en}, one has to combine the Hagedorn wavepackets for the $|k|$th eigenspace. We set
$\Phi^\eps_{|k|}=\Phi^\eps_{|k|}[q,p,Q,P]$, 
$$
\Phi^\eps_{|k|}(x) = (\phi^\eps_{l_1}(x),\ldots,\phi^\eps_{l_d}(x))_{l_1+\cdots+l_d=|k|}\in\C^{N},\qquad x\in\R^d,
$$
and move to the associated matrix-valued Wigner transform
$$
\W^\eps(\Phi^\eps_{|k|},\Phi^\eps_{|k|})(x,\xi) = 
(2\pi\eps)^{-d} \int_{\R^d} \overline\Phi^\eps_{|k|}(x+\tfrac{y}{2})\Phi^\eps_{|k|}(x-\tfrac{y}{2})^T \e^{iy^T \xi/\eps} dy\in \C^{N\times N}.
$$

\begin{corollary} 
Let $\eps>0$, $q,p\in\R^d$ and $Q,P\in\C^{d\times d}$ satisfy (\ref{eq:QP}). Then,   
the Hagedorn wavepackets' vector for the $|k|$th eigenspace, $\Phi^\eps_{|k|}=\Phi^\eps_{|k|}[q,p,Q,P]$, $k\in\N^d$, satisfies
$$
\tr\left(\W^\eps(\Phi^\eps_{|k|},\Phi^\eps_{|k|})(x,\xi)\right) = \frac{(-1)^{|k|}}{(\pi\eps)^d} \e^{-\tfrac{1}{\eps}|z|^2} L^{(d-1)}_{|k|}(\tfrac{2}{\eps}|z|^2) 
$$
with $z=-i\left(P^T(x-q)-Q^T(\xi-p)\right)$ for $(x,\xi)\in\R^{2d}$. 
\end{corollary}

\begin{proof}
\begin{eqnarray*}
\lefteqn{\tr\left(\W^\eps(\Phi^\eps_{|k|},\Phi^\eps_{|k|})(x,\xi)\right)}\\ 
&=& 
\frac{(-1)^{|k|}}{(\pi\eps)^d} \e^{-\tfrac{1}{\eps}|z|^2} 
\sum_{l_1+\cdots+l_d=|k|} L_{l_1}^{(0)}(\tfrac{2}{\eps}|z_1|^2)\cdots L_{l_d}^{(0)}(\tfrac{2}{\eps}|z_d|^2)\\
&=&
\frac{(-1)^{|k|}}{(\pi\eps)^d} \e^{-\tfrac{1}{\eps}|z|^2} L^{(d-1)}_{|k|}(\tfrac{2}{\eps}|z|^2)
\end{eqnarray*}
by the Laguerre polynomials' summation theorem \cite[\S119]{Ra}. 
\end{proof}

\subsection{Three-term recurrence relation}

The Hermite functions' ladder operators can be translated to the Wigner function level, see \cite[Theorem 1.3.3]{Tn}. The same is possible for the Hagedorn wavepackets and provides a useful three-term recurrence relation in phase space. 

\begin{theorem}[Phase space ladder]\label{theo:three}
Let $\eps>0$, $q,p\in\R^{d}$, and $Q,P\in\C^{d\times d}$ satisfy~(\ref{eq:QP}). Then, the Wigner transform $\W^\eps_{kl}=\W^\eps(\phi^\eps_k,\phi^\eps_l)$ of the $k$th and $l$th Hagedorn wavepackets $\phi^\eps_k=\phi^\eps_k[q,p,Q,P]$ and $\phi^\eps_l=\phi^\eps_l[q,p,Q,P]$, $k,l\in\N^d$, satisfies
\begin{eqnarray*}
K_j^\dagger\W^\eps_{kl} &=& 2\sqrt{k_{j}+1} \W^\eps_{k+e_j,l},\qquad
K_j\W^\eps_{kl} = 2\sqrt{k_j} \W^\eps_{k-e_j,l},\\
L_j^\dagger\W^\eps_{kl} &=& 2\sqrt{l_{j}+1} \W^\eps_{k,l+e_j},\qquad
L_j\W^\eps_{kl} = 2\sqrt{l_j} \W^\eps_{k,l-e_j}
\end{eqnarray*}
for $j=1,\ldots,d$ with 
\begin{eqnarray*}
K^\dagger &=& \tfrac{i}{\sqrt{2\eps}} \left(P^T(2(x-q)+(-i\eps\nabla_\xi)) - Q^T(2(\xi-p)-(-i\eps\nabla_x))\right),\\
K &=& -\tfrac{i}{\sqrt{2\eps}} \left(P^*(2(x-q)+(-i\eps\nabla_\xi)) - Q^*(2(\xi-p)-(-i\eps\nabla_x))\right)
\end{eqnarray*}
and
\begin{eqnarray*}
L^\dagger &=& \tfrac{i}{\sqrt{2\eps}} \left(P^*(2(x-q)-(-i\eps\nabla_\xi)) - Q^*(2(\xi-p)+(-i\eps\nabla_x))\right),\\
L &=& -\tfrac{i}{\sqrt{2\eps}} \left(P^T(2(x-q)-(-i\eps\nabla_\xi)) - Q^T(2(\xi-p)+(-i\eps\nabla_x))\right).
\end{eqnarray*}
\end{theorem}

\begin{proof}
We compute
\begin{eqnarray*}
(-i\eps\nabla_x)\W^\eps_{kl}(x,\xi) &=&-(2\pi\eps)^{-d}\int_{\R^d} 
\overline{(-i\eps\nabla_x-p)\phi^\eps_k(x+\tfrac{y}{2})}\phi^\eps_l(x-\tfrac{y}{2})\e^{\tfrac{i}{\eps}\xi^Ty} dy\\
&&+(2\pi\eps)^{-d}\int_{\R^d} 
\overline{\phi^\eps_k(x+\tfrac{y}{2})}(-i\eps\nabla_x-p)\phi^\eps_l(x-\tfrac{y}{2})\e^{\tfrac{i}{\eps}\xi^Ty} dy
\end{eqnarray*}
and
\begin{eqnarray*}
(-i\eps\nabla_\xi) \W^\eps_{kl}(x,\xi) &=& (2\pi\eps)^{-d}\int_{\R^d} 
\overline{(x+\tfrac{y}{2}-q)\phi^\eps_k(x+\tfrac{y}{2})}\phi^\eps_l(x-\tfrac{y}{2})\e^{\tfrac{i}{\eps}\xi^Ty} dy\\
&& -(2\pi\eps)^{-d}\int_{\R^d} \overline{\phi^\eps_k(x+\tfrac{y}{2})}(x-\tfrac{y}{2}-q)\phi^\eps_l(x-\tfrac{y}{2})\e^{\tfrac{i}{\eps}\xi^Ty} dy.
\end{eqnarray*}
Since 
\begin{eqnarray*}
A^\dagger &=& \tfrac{i}{\sqrt{2\eps}} \left(P^* (x-q) - Q^*(-i\eps\nabla_x - p)\right),\\
A &=& -\tfrac{i}{\sqrt{2\eps}} \left(P^T(x-q) - Q^T(-i\eps\nabla_x- p)\right),
\end{eqnarray*}
we obtain
\begin{eqnarray*}
\lefteqn{-\tfrac{i}{\sqrt{2\eps}} \left(P^*(-i\eps\nabla_\xi)+Q^* (-i\eps\nabla_x)\right)\W^\eps_{kl}(x,\xi) =}\\ 
&&
(2\pi\eps)^{-d} \int_{\R^d} 
\left(\overline{A\phi^\eps_k(x+\tfrac{y}{2})}\phi^\eps_l(x-\tfrac{y}{2})+\overline{\phi^\eps_k(x+\tfrac{y}{2})}A^\dagger\phi^\eps_l(x-\tfrac{y}{2})\right)\e^{\tfrac{i}{\eps}\xi^Ty} dy
\end{eqnarray*}
and
\begin{eqnarray*}
\lefteqn{\tfrac{i}{\sqrt{2\eps}} \left(P^T(-i\eps\nabla_\xi)+Q^T (-i\eps\nabla_x)\right)\W^\eps_{kl}(x,\xi) =}\\ 
&&
(2\pi\eps)^{-d} \int_{\R^d} 
\left(\overline{A^\dagger\phi^\eps_k(x+\tfrac{y}{2})}\phi^\eps_l(x-\tfrac{y}{2})+\overline{\phi^\eps_k(x+\tfrac{y}{2})}A\phi^\eps_l(x-\tfrac{y}{2})\right)\e^{\tfrac{i}{\eps}\xi^Ty} dy.
\end{eqnarray*}
Using $A^\dagger_j\phi^\eps_k=\sqrt{k_j+1}\phi^\eps_{k+e_j}$ and $A_j\phi^\eps_k=\sqrt{k_j}\phi^\eps_{k-e_j}$, we arrive at
\begin{eqnarray*}
-\tfrac{i}{\sqrt{2\eps}} \left(P^*(-i\eps\nabla_\xi)+Q^* (-i\eps\nabla_x)\right)_j\W^\eps_{kl} &=& \sqrt{k_j}\W^\eps_{k-e_j,l} + \sqrt{l_j+1} \W^\eps_{k,l+e_j},\\
\tfrac{i}{\sqrt{2\eps}} \left(P^T(-i\eps\nabla_\xi)+Q^T (-i\eps\nabla_x)\right)_j\W^\eps_{kl} &=& \sqrt{k_j+1}\W^\eps_{k+e_j,l} + \sqrt{l_j} \W^\eps_{k,l-e_j}.
\end{eqnarray*}
Moreover, 
\begin{eqnarray*}
2(x-q)\W^\eps_{kl}(x,\xi) &=&  (2\pi\eps)^{-d}\int_{\R^d} 
\overline{(x+\tfrac{y}{2}-q)\phi^\eps_k(x+\tfrac{y}{2})}\phi^\eps_l(x-\tfrac{y}{2})\e^{\tfrac{i}{\eps}\xi^Ty} dy\\
&+&(2\pi\eps)^{-d}\int_{\R^d}\overline{\phi^\eps_k(x+\tfrac{y}{2})}(x-\tfrac{y}{2}-q)\phi^\eps_l(x-\tfrac{y}{2})\e^{\tfrac{i}{\eps}\xi^Ty} dy.
\end{eqnarray*}
Since $\xi\e^{\tfrac{i}{\eps}\xi^Ty}=-i\eps\nabla_y\e^{\tfrac{i}{\eps}\xi^Ty}$, an integration by parts yields
\begin{eqnarray*}
2(\xi-p)\W^\eps_{kl}(x,\xi) &=& (2\pi\eps)^{-d}\int_{\R^d} 
\overline{(-i\eps\nabla_x-p)\phi^\eps_k(x+\tfrac{y}{2})}\phi^\eps_l(x-\tfrac{y}{2})\e^{\tfrac{i}{\eps}\xi^Ty} dy\\
&+&(2\pi\eps)^{-d}\int_{\R^d}\overline{\phi^\eps_k(x+\tfrac{y}{2})}(-i\eps\nabla_x-p)\phi^\eps_l(x-\tfrac{y}{2})\e^{\tfrac{i}{\eps}\xi^Ty} dy,
\end{eqnarray*}
Consequently,
\begin{eqnarray*}
\lefteqn{\tfrac{i}{\sqrt{2\eps}} \left(2P^*(x-q)-2Q^* (\xi-p)\right)\W^\eps_{kl}(x,\xi) =}\\ 
&&
(2\pi\eps)^{-d} \int_{\R^d} 
\left(-\overline{A\phi^\eps_k(x+\tfrac{y}{2})}\phi^\eps_l(x-\tfrac{y}{2})+\overline{\phi^\eps_k(x+\tfrac{y}{2})}A^\dagger\phi^\eps_l(x-\tfrac{y}{2})\right)\e^{\tfrac{i}{\eps}\xi^Ty} dy
\end{eqnarray*}
and
\begin{eqnarray*}
\lefteqn{-\tfrac{i}{\sqrt{2\eps}} \left(2P^T(x-q)-2Q^T (\xi-p)\right)\W^\eps_{kl}(x,\xi) =}\\ 
&&
(2\pi\eps)^{-d} \int_{\R^d} 
\left(-\overline{A^\dagger\phi^\eps_k(x+\tfrac{y}{2})}\phi^\eps_l(x-\tfrac{y}{2})+\overline{\phi^\eps_k(x+\tfrac{y}{2})}A\phi^\eps_l(x-\tfrac{y}{2})\right)\e^{\tfrac{i}{\eps}\xi^Ty} dy
\end{eqnarray*}
This implies
\begin{eqnarray*}
\tfrac{i}{\sqrt{2\eps}} \left(2P^*(x-q)-2Q^* (\xi-p)\right)_j\W^\eps_{kl} &=& -\sqrt{k_j}\W^\eps_{k-e_j,l} + \sqrt{l_j+1} \W^\eps_{k,l+e_j},\\
-\tfrac{i}{\sqrt{2\eps}} \left(2P^T(x-q)-2Q^T (\xi-p)\right)_j\W^\eps_{kl} &=& -\sqrt{k_j+1}\W^\eps_{k+e_j,l} + \sqrt{l_j} \W^\eps_{k,l-e_j}.
\end{eqnarray*}
Adding and subtracting the above raising and lowering identities for $\W^\eps_{kl}$ gives the operators $K^\dagger$, $K$, $L^\dagger$, and $L$.
\end{proof}

The phase space ladder allows to reformulate the Hagedorn wavepackets' three-term recurrence relation (\ref{eq:3t}) for the Wigner transform. 
\begin{corollary}[Three-term recurrence]
\label{cor:3t}
Let $\eps>0$, $q,p\in\R^{d}$, and $Q,P\in\C^{d\times d}$ be matrices satisfying~(\ref{eq:QP}). Then, the Wigner transform $\W^\eps_{kl}=\W^\eps(\phi^\eps_k,\phi^\eps_l)$ of the $k$th and the $l$th Hagedorn wavepackets $\phi^\eps_k=\phi^\eps_k[q,p,Q,P]$ and $\phi^\eps_l=\phi^\eps_l[q,p,Q,P]$, $k,l\in\N^d$, satisfies
\begin{eqnarray*}
\left(\sqrt{k_j+1}\W^\eps_{k+e_j,l}(x,\xi)\right)_{j=1}^d = -\sqrt{\tfrac{2}{\eps}}\,z\,\W^\eps_{kl}(x,\xi) + \left(\sqrt{l_j}\W^\eps_{k,l-e_j}(x,\xi)\right)_{j=1}^d\\
\left(\sqrt{l_j+1}\W^\eps_{k,l+e_j}(x,\xi)\right)_{j=1}^d = \sqrt{\tfrac{2}{\eps}}\,\overline{z}\,\W^\eps_{kl}(x,\xi) + \left(\sqrt{k_j}\W^\eps_{k-e_j,l}(x,\xi)\right)_{j=1}^d
\end{eqnarray*}
with $z=-i\left(P^T(x-q)-Q^T(\xi-p)\right)$ for $(x,\xi)\in\R^{2d}$.
\end{corollary}

\begin{proof}
We observe that
\begin{eqnarray*}
\left(K^\dagger_j-L_j\right)\W^\eps_{kl} 
&=& 
\tfrac{2i}{\sqrt{2\eps}}\left(2P^T(x-q)-2Q^T(\xi-p)\right)_j\W^\eps_{kl}\\
&=& 
2\sqrt{k_j+1}\W^\eps_{k+e_j,l} - 2\sqrt{l_j}\W^\eps_{k,l-e_j},\\
\left(L^\dagger_j-K_j\right)\W^\eps_{kl} 
&=& 
\tfrac{2i}{\sqrt{2\eps}}\left(2P^*(x-q)-2Q^*(\xi-p)\right)_j\W^\eps_{kl}\\
&=& 
2\sqrt{l_j+1}\W^\eps_{k,l+e_j} - 2\sqrt{k_j}\W^\eps_{k-e_j,l}.
\end{eqnarray*}
\end{proof}

The numerical computation of the Wigner transform of a Schwartz function $\psi:\R^d\to\C$ for a set of phase space points $\zeta_1,\ldots,\zeta_N\in\R^{2d}$ is notoriously difficult, since a direct approach poses the numerical quadrature of $N$ Fourier integrals in possibly high dimensions, 
$$
\W^\eps(\psi,\psi)(x,\xi) =  
(2\pi\eps)^{-d} \int_{\R^d} \overline\psi(x+\tfrac{y}{2})\psi(x-\tfrac{y}{2}) \e^{iy^T \xi/\eps} dy
$$
for $(x,\xi)\in\{\zeta_1,\ldots,\zeta_N\}$. The three-term recurrence of Corollary~\ref{cor:3t} might provide an efficient alternative, especially when a large number $N\gg1$ of evaluations is required: Choosing suitable $q,p\in\R^d$, $Q,P\in\C^{d\times d}$ satisfying (\ref{eq:QP}), $K\in\N$, and the hyperbolic multi-index set
$$
\K=\{k\in\N^d\mid \prod_{j=1}^d (1+k_j)\le K\}.
$$ 
Then, one approximates
$$
\psi \approx \sum_{k\in\K} c_k\phi^\eps_k, \qquad c_k = \int_{\R^d} \overline{\phi^\eps_k}(y) \psi(y) dy
$$
with $\phi^\eps_k=\phi^\eps_k[q,p,Q,P]$ for $k\in\K$. With literally the same proof as for the Hermite functions \cite[Theorem 1.5]{Lu}, the approximation error can be estimated as 
$$
\big\|\psi-\sum_{k\in\K}c_k\phi^\eps_k\big\|_{L^2(\R^d)} \le C_{s,d} K^{-s/2}\max_{\|\sigma\|_\infty\le s}\left\|A[q,p,Q,P]^\sigma \psi\right\|_{L^2(\R^d)}
$$
for fixed $s\in\N$ and every Schwartz function $\psi:\R^d\to\C$. From the Wigner transform's bilinearity and the orthogonality relation (\ref{eq:orth}) we deduce
\begin{eqnarray*}
\lefteqn{\| \W^\eps(\phi_1) -\W^\eps(\phi_2)\|_{L^2(\R^{2d})}}\\ 
&\le& 
\| \W^\eps(\phi_1,\phi_1-\phi_2)\|_{L^2(\R^{2d})}+\|\W^\eps(\phi_1-\phi_2,\phi_2)\|_{L^2(\R^{2d})}\\
&=&
\|\phi_1-\phi_2\|_{L^2(\R^d)} \left(\|\phi_1\|_{L^2(\R^d)} + \|\phi_2\|_{L^2(\R^d)}\right) 
\end{eqnarray*}
for all Schwartz functions $\phi_1,\phi_2:\R^d\to\C$. Therefore, evaluating the Wigner function of $\psi$ via 
$$
\W^\eps(\psi,\psi) \approx \sum_{k,l\in\K} \overline{c}_k c_l \W^\eps(\phi^\eps_k,\phi^\eps_l),
$$
together with the three-term recurrence relation of Corollary~\ref{cor:3t}, we inherit the approximation accuracy of $O(K^{-s/2})$.

\subsection{FBI transform}

Our general findings for the FBI transform of the Hagedorn wavepackets are less beautiful than those for the Wigner function.  

\begin{proposition}[FBI transform] Let $\eps>0$, $q,p\in\R^d$ and $Q,P\in\C^{d\times d}$ satisfy (\ref{eq:QP}). Let $\phi^\eps_k=\phi^\eps_k[q,p,Q,P]$ be the $k$th Hagedorn wavepacket. Then the scaled FBI transform is 
\begin{eqnarray*}
\lefteqn{\T^\eps(\phi^\eps_k)(x,\xi) = \frac{\e^{\tfrac{i}{\eps}\xi^T(x-\tfrac12q)}\,\e^{-\tfrac{i}{2\eps}p^Tx}\det(Q)^{-1/2}}{(\pi\eps)^d\sqrt{2^{|k|+d}k!}}}\\
&& \e^{-\tfrac{1}{2\eps}|x-q|^2 +\tfrac{1}{2\eps}w^T(\Id-iC)w} \sum_{\nu\le k} \binom{k}{\nu} \left(\tfrac{2}{\sqrt\eps}{Q^{-1}}w\right)^{k-\nu}c_\nu
\end{eqnarray*}
with $w=(\Id-iC)^{-1}((x-q)-i(\xi-p))$ for $(x,\xi)\in\R^{2d}$. In particular, if $C=i\Id$, then
$$
\T^\eps(\phi^\eps_k)(x,\xi) = \frac{\e^{\tfrac{i}{2\eps}(\xi^Tx-p^Tq)}\det(Q)^{1/2}}{(\pi\eps)^{d/2}\sqrt{2^{|k|+d}k!}} \e^{-\tfrac{1}{4\eps}|z|^2}  \left(\tfrac{1}{\sqrt\eps}Q^{-1}z\right)^{k}
$$
and
$$
\H^\eps(\phi^\eps_k,\phi^\eps_k)(x,\xi) = \frac{1}{(\pi\eps)^{d} 2^{|k|+d}k!} \e^{-\tfrac{1}{2\eps}|z|^2} \left|\tfrac{1}{\sqrt\eps}z\right|^{2k}
$$
with $z=(x-q)-i(\xi-p)$ for $(x,\xi)\in\R^{2d}$. 
\end{proposition}

\begin{proof} 
We start for $(x,\xi)=(0,0)$ and obtain
$$
\T^\eps(\phi^\eps_k)(0,0) = \frac{\det(Q)^{-1/2}}{(\pi\eps)^d\sqrt{2^{|k|+d}k!}}\int_{\R^d} p^\eps_k(y) \e^{\tfrac{i}{2\eps}(y-q)^T C(y-q)+\tfrac{i}{\eps}p^T(y-q)-\tfrac{1}{2\eps}|y|^2} dy.
$$
We compute
\begin{eqnarray*}
\lefteqn{\tfrac{i}{2\eps}(y-q)^T C(y-q)+\tfrac{i}{\eps}p^T(y-q)-\tfrac{1}{2\eps}|y|^2}\\
&=& 
\tfrac{1}{2\eps}(y-q)^T (iC+\Id)(y-q)+\tfrac1\eps(ip-q)^T(y-q)-\tfrac{1}{2\eps}|q|^2\\
&=&
-\tfrac{1}{2\eps}(y-q-w)^T(\Id- iC)(y-q-w)+\tfrac{1}{2\eps}w^T(\Id-iC)w-\tfrac{1}{2\eps}|q|^2 
\end{eqnarray*}
with $w=(\Id-iC)^{-1}(ip-q)$. Therefore, 
\begin{eqnarray*}
\T^\eps(\phi^\eps_k)(0,0) &=& \frac{\det(Q)^{-1/2}}{(\pi\eps)^d\sqrt{2^{|k|+d}k!}} \e^{-\tfrac{1}{2\eps}|q|^2+\tfrac{1}{2\eps}w^T(\Id-iC)w} \\
&&
\int_{\R^d} p^\eps_k(y+w) \e^{-\tfrac{1}{2\eps}(y-q)^T(\Id-iC)(y-q)} dy.
\end{eqnarray*}
Let $M=\frac12(\Id-i\overline C)$. Then, $M=M^T$, $\Im(C)+\Re(M) = \frac12\Id+\frac12\Im(C)>0$, and 
$$
(\Id+Q^*MQ)^* = \Id + \tfrac12 Q^*(\Id+iC)Q = \Id + \tfrac12 Q^*Q + \tfrac{i}{2} Q^*P = \tfrac12 (Q^*+iP^*)Q,
$$ 
where we have used the matrix property (\ref{eq:QP}). Since
$$
(Q+iP)^*(Q+iP) = (Q-iP)^*(Q-iP) - 4\Id,
$$
the matrices $Q-iP$ and $\Id+Q^*MQ$ are invertible. Moreover,
$$
\Im(C)+M=\tfrac12(\Id-iC),
$$
and we have by Proposition~\ref{prop:rod}
\begin{eqnarray*}
\lefteqn{\T^\eps(\phi^\eps_k)(0,0) =}\\
&&
 \frac{\det(Q)^{-1/2}}{(\pi\eps)^d\sqrt{2^{|k|+d}k!}} \e^{-\tfrac{1}{2\eps}|q|^2+\tfrac{1}{2\eps}w^T(\Id-iC)w}
\sum_{\nu\le k} \binom{k}{\nu} \left(\tfrac{2}{\sqrt\eps}{Q^{-1}}w\right)^{k-\nu}c_\nu.
\end{eqnarray*}
For arbitrary $(x,\xi)\in\R^{2d}$, we use Lemma~\ref{lem:mt} and obtain 

$$
T_{x,\xi}^{-1}\phi^\eps_k = \e^{-\tfrac{i}{2\eps}p^Tq} \e^{\tfrac{i}{2\eps}(p-\xi)^T(q-x)}\phi^\eps_k[q-x,p-\xi,Q,P]
$$
as well as
\begin{eqnarray*}
\T^\eps(\phi_k^\eps)(x,\xi) 
&=& \e^{\frac{i}{2\eps}\xi^T x} \,\T^\eps(T_{x,\xi}^{-1}\phi^\eps_k)(0,0)\\ 
&=& \e^{\tfrac{i}{\eps}\xi^T(x-\tfrac12q)}\,\e^{-\tfrac{i}{2\eps}p^Tx} \,\T^\eps(\phi^\eps_k[q-x,p-\xi,Q,P])(0,0).
\end{eqnarray*}
In the special case $C=i\Id$, we have $M=0$, $w=\tfrac12((x-q)-i(\xi-p))$ and $c_\nu=0$ for $\nu\neq0$. Consequently, 
$$
\T^\eps(\phi^\eps_k)(x,\xi) = \frac{ \e^{\tfrac{i}{2\eps}(\xi^Tx-p^Tq)}\det(Q)^{1/2}}{(\pi\eps)^{d/2}\sqrt{2^{|k|+d}k!}} \e^{-\tfrac{1}{4\eps}|z|^2}  \left(\tfrac{1}{\sqrt\eps}{Q^{-1}}z\right)^{k}
$$
with $z=(x-q)-i(\xi-p)=2w$.
\end{proof}

\section{Generalized coherent states}
\label{sec:genstat}

The Hagedorn wavepackets coexist in the literature under the name of generalized squeezed states \cite{C} or generalized coherent states \cite[\S3.4]{CR}. A generalized coherent state is constructed using the squeezing operators
$$
D_B = \exp\!\left(\tfrac{1}{2}\left((a^\dagger)^T Ba^\dagger - a^T B^*a\right)\right),\qquad B=B^T\in\C^{d\times d}, 
$$
where 
$$
a = \tfrac{1}{\sqrt{2\eps}}(x+\eps\nabla_x) = A[0,0,\Id,i\Id],\qquad
a^\dagger = \tfrac{1}{\sqrt{2\eps}}(x-\eps\nabla_x) = A^\dagger[0,0,\Id,i\Id]
$$
are the ladder operators of the $\eps$-scaled harmonic oscillator $\frac12(-\eps^2\Delta_x+|x|^2)$. 
Let $\psi^\eps_k$ be the $k$th normalized eigenstate of this oscillator, that is, 
$$
\psi^\eps_k(x) = (\pi\eps)^{-d}\e^{-|x|^2/(2\eps)}\frac{1}{\sqrt{2^{|k|}k!}} \prod_{j=1}^d h_{k_j}\!\left(\tfrac{x_j}{\sqrt\eps}\right),
\qquad x\in\R^d. 
$$
Then, the generalized coherent states for $q,p\in\R^d$, $B=B^T\in\C^{d\times d}$ are defined by 
translating and squeezing, that is, 
$$
\widetilde\phi^\eps_k[q,p,B] = T_{q,p} D_B \,\psi^\eps_k,\qquad k\in\N^d.
$$
where $T_{q,p}$ denotes the Heisenberg--Weyl translation operator used in Lemma~\ref{lem:mt}.

We note, that the complex symmetric matrix $B=B^T\in\C^{d\times d}$ of the squeezing operator is typically constructed 
from a complex symmetric matrix 
$$
W=W^T\in\C^{d\times d}\quad\mbox{with}\quad W^*W<\Id.
$$ 
If $W=U(W^*W)^{1/2}$ is a polar decomposition of $W$ with $U\in\C^{d\times d}$ unitary, then 
$B=U\artanh(W^*W)^{1/2}$. In this case, 
\begin{equation}\label{eq:squeezeW}
D_B a^\dagger D_B^{-1} = (\Id-W^*W)^{-1/2}(a^\dagger-W^*a).
\end{equation}
see \cite[Lemma~25]{CR}.

\subsection{Relation to the Hagedorn wavepackets}

The explicit action of the squeezing operators on the Hagedorn ladders is a bit more intricate than that of the translation operators previously discussed in Lemma~\ref{lem:mt}.
\begin{proposition}[Hagedorn wavepackets and generalized coherent states] Let $W=W^T\in\C^{d\times d}$ satisfy $W^*W<\Id$. Then,
$$
Q = (\Id+W)(\Id-W^*W)^{-1/2},\qquad P = i(\Id-W)(\Id-W^*W)^{-1/2}
$$
define matrices satisfying (\ref{eq:QP}), while the squeezing operator $D(B)$ associated with~$W$ satisfies
\begin{equation}\label{eq:squeeze}
D_B^{-1} A^\dagger[0,0,Q,P] D_B = a^\dagger,\qquad
D_B^{-1} A[0,0,Q,P] D_B = a.
\end{equation}
Moreover, for all $q,p\in\R^d$ there exists $c\in\C$, $|c|=1$, such that 
\begin{equation}\label{eq:hw=gcs}
\widetilde\phi^\eps_k[q,p,B]= c \,\phi^\eps_k[q,p,Q,P],\qquad k\in\N^d.
\end{equation}
Conversely, let $Q,P\in\C^{d\times d}$ satisfy (\ref{eq:QP}). Then, 
$W=(Q+iP)(Q-iP)^{-1}$ 
is a complex symmetric matrix with $W^*W<\Id$, and the associated squeezing operator $D(B)$ satisfies
$$
D_B^{-1}\, A^\dagger[0,0,QV,PV] D_B = a^\dagger,\qquad 
D_B^{-1} A[0,0,QV,PV] D_B = a,
$$
where the unitary matrix $V\in\C^{d\times d}$ results from the polar decomposition of 
$Q-iP=|Q-iP|V^*$, $|Q-iP|^2=(Q-iP)(Q-iP)^*$. Moreover, for all $q,p\in\R^d$ there exists $c\in\C$, $|c|=1$, such that 
$$
\widetilde\phi^\eps_k[q,p,B]=c\,\phi^\eps_k[q,p,QV,PV],\qquad k\in\N^d.
$$ 
\end{proposition}
\begin{proof} Let $W=W^T$ satisfy $W^*W<\Id$. Then,
\begin{eqnarray*}
Q^TP - P^TQ &=& i(\Id-W^*W)^{-T/2}\left((\Id+W)(\Id-W)\right.\\
&& \left.-(\Id-W)(\Id+W)\right)(\Id-W^*W)^{-1/2} = 0,\\
Q^*P - P^*Q &=& i(\Id-W^*W)^{-1/2}\left((\Id+W^*)(\Id-W)\right.\\
&& \left.+(\Id-W^*)(\Id+W)\right)(\Id-W^*W)^{-1/2} = 2i\Id.
\end{eqnarray*}
Therefore, by relation~\ref{eq:squeezeW},
\begin{eqnarray*}
D_B a^\dagger D_B^{-1} &=& 
\tfrac{i}{\sqrt{2\eps}}(\Id-W^*W)^{-1/2} \left(-i(\Id-W^*)x- (\Id+W^*)(-i\eps\nabla_x)\right)\\
&=& A^\dagger[0,0,Q,P]
\end{eqnarray*}
and $D_B a D_B^{-1}=A[0,0,Q,P]$. By \cite[Proposition~36]{CR},
$$
\widetilde\phi^\eps_0[0,0,B](x) = D_B \psi^\eps_0(x) = a_\Gamma \exp\!\left(\tfrac{i}{2\eps} x^T \Gamma x\right),\qquad x\in\R^d,
$$
where $a_\Gamma\in\C$ is a normalizing constant and $\Gamma=i(\Id-W)(\Id+W)^{-1}=PQ^{-1}$. Therefore, 
\begin{eqnarray*}
\widetilde\phi^\eps_0[q,p,B](x) &=& T_{q,p}\widetilde\phi^\eps_0[0,0,B](x)\\ 
&=& a_\Gamma\exp(-\tfrac{i}{2\eps}q^Tp)\exp(\tfrac{i}{\eps}x^Tp)
\exp(\tfrac{i}{2\eps} (x-q)^T PQ^{-1} (x-q))
\end{eqnarray*}
for $x\in\R^d$. Hence, there exists $c\in\C$, $|c|=1$ such that $\widetilde\phi^\eps_0[q,p,B]=c\,\phi^\eps_0[q,p,Q,P]$. 
For proving~(\ref{eq:hw=gcs}) we argue inductively. We have
\begin{eqnarray*}
\widetilde\phi^\eps_{k+e_j}[q,p,B] &=& 
\tfrac{1}{\sqrt{k_j+1}} T_{q,p} D_B a_j^\dagger \psi^\eps_k\\
&=& \tfrac{1}{\sqrt{k_j+1}} A_j^\dagger[q,p,Q,P] T_{q,p} D_B \psi^\eps_k\\
&=& \tfrac{c}{\sqrt{k_j+1}} A_j^\dagger[q,p,Q,P] \phi^\eps_k[q,p,Q,P]\\
&=& c\,\phi^\eps_{k+e_j}[q,p,Q,P]
\end{eqnarray*}
for all $q,p\in\R^d$ and $k\in\N^d$.

Conversely, let $Q,P$ satisfy (\ref{eq:QP}). Then, by Remark~\ref{rem:sympl} and \cite[Lemma~23]{CR}, $W=(Q+iP)(Q-iP)^{-1}$ satisfies $W=W^T$ and $W^*W = \Id - 4((Q-iP)(Q-iP)^*)^{-1}$. Therefore, 
\begin{eqnarray*}
(\Id+W)(\Id-W^*W)^{-1/2} &=& \tfrac12 \left((Q-iP)+(Q+iP)\right)(Q-iP)^{-1} |Q-iP|\\
&=& QV,\\
i(\Id-W)(\Id-W^*W)^{-1/2} &=& \tfrac{i}{2} \left((Q-iP)-(Q+iP)\right)(Q-iP)^{-1} |Q-iP|\\
&=& PV,
\end{eqnarray*}
and $D_B a^\dagger D_B^{-1}= A^\dagger[0,0,QV,PV]$. As before, for all $q,p\in\R^d$ there exists 
$C\in\C$, $|c|=1$ such that $\widetilde\phi^\eps_k[q,p,B]=c\,\phi^\eps_k[q,p,QV,PV]$ for all $k\in\N^d$.
\end{proof}

\subsection{Generalized coherent states in phase space}

With any symplectic matrix $F\in\R^{2d\times 2d}$ one can associate a unitary operator $R(F)$ on $L^2(\R^d)$ satisfying
$$
R_F^{-1}\, T_{q,p}\, R_F = T_{F^{-1}(q,p)},\qquad (q,p)\in\R^{2d}.
$$
The operator $R_F$ is uniquely determined up to a multiplicative constant $\lambda\in\C$, $|\lambda|=1$, and is called the metaplectic transformation of the symplectic matrix $F$. We have
$$
\W^\eps(R_F\phi,R_F\psi)(x,\xi) = W^\eps(\phi,\psi)(F^{-1}(x,\xi)),\qquad (x,\xi)\in\R^{2d},
$$
see \cite[Corollary~217]{Go} or \cite[Proposition~18 and \S3.3]{CR}. Since the squeezing operator is a metaplectic transformation, there is an alternative way in addition to the sum rules for computing the Hagedorn wavepackets' Wigner transform.
\begin{theorem}[Wigner transform, 2nd proof] Let $q,p\in\R^d$. Let $W=W^T\in\C^{d\times d}$ with $W^*W<\Id$ and 
$$
Q = (\Id+W)(\Id-W^*W)^{-1/2},\qquad P = i(\Id-W)(\Id-W^*W)^{-1/2}.
$$ 
Then, the scaled Wigner function of the $k$th and the $l$th Hagedorn wavepacket $\phi_k^\eps=\phi^\eps_k[q,p,Q,P]$ and $\phi_l^\eps=\phi^\eps_l[q,p,Q,P]$, $k,l\in\N^d$, satisfies (\ref{eq:hagwig}), that is,
$$
\W^\eps(\phi_k^\eps,\phi_l^\eps)(x,\xi) = 
(\pi\eps)^{-d} \e^{-\tfrac{1}{\eps}|z|^2} \frac{(-1)^{|l|}}{\sqrt{2^{|k|+|l|}k! l!}} \prod_{j=1}^d  
{\mathcal L}_{k_j,l_j}(\tfrac{1}{\sqrt\eps}z_j)
$$
with $z(x,\xi)=-i(P^T(x-q)-Q^T(\xi-p))$ for $(x,\xi)\in\R^{2d}$.
\end{theorem}
\begin{proof} By the relations in (\ref{eq:squeeze}),
\begin{eqnarray*}
\lefteqn{D_B x D_B^{-1}}\\ 
&=& \sqrt{\tfrac\eps2} D_B (a + a^\dagger) D_B^{-1} 
=\sqrt{\tfrac\eps2} \left( A[0,0,Q,P] + A^\dagger[0,0,Q,P]\right)\\
&=&  \Im(P)^T x - \Im(Q)^T (-i\eps\nabla_x)
\end{eqnarray*}
and
\begin{eqnarray*}
\lefteqn{D_B (-i\eps\nabla_x) D_B^{-1}}\\ 
&=& i\,\sqrt{\tfrac{\eps}{2}} D_B (a-a^\dagger) D_B^{-1}
= i\,\sqrt{\tfrac\eps2} \left( A[0,0,Q,P] - A^\dagger[0,0,Q,P]\right)\\
&=&  -\Re(P)^T x + \Re(Q)^T (-i\eps\nabla_x).
\end{eqnarray*}
Hence, 
$$
D_B\begin{pmatrix}x\\-i\eps\nabla_x\end{pmatrix}D_B^{-1}=F^{-1}\begin{pmatrix}x\\-i\eps\nabla_x\end{pmatrix}
$$ 
with $F^{-1}=-JF^TJ$ the inverse matrix of the symplectic matrix 
$$
F = \begin{pmatrix}\Re(Q)&\Im(Q)\\ \Re(P)&\Im(P)\end{pmatrix}.
$$
This implies $D_B^{-1}=R_{F^{-1}}$ and $D_B=R_{F}$,  see \cite[\S3.3]{CR}. Therefore, using Lemma~\ref{lem:mt},
\begin{eqnarray*}
\W^\eps(\phi^\eps_k,\phi^\eps_l)(x,\xi) &=& \W^\eps(T_{q,p}D_B\psi^\eps_k,T_{q,p}D_B\psi^\eps_l)(x,\xi)\\
&=& \W^\eps(\psi^\eps_k,\psi^\eps_l)(F^{-1}(x-q,\xi-p))
\end{eqnarray*}
for all $(x,\xi)\in\R^{2d}$. Since the real and the imaginary part of the complex vector $z(x,\xi)=-i(P^T(x-q)-Q^T(\xi-p))$ can be written as
$$
\begin{pmatrix}\Re(z)\\ \Im(z)\end{pmatrix} = 
\begin{pmatrix}\Im(P)^T&-\Im(Q)^T\\ -\Re(P)^T&\Re(Q)^T\end{pmatrix}\begin{pmatrix}x-q\\ \xi-p\end{pmatrix} = 
F^{-1}\begin{pmatrix}x-q\\ \xi-p\end{pmatrix},
$$
we have 
$$
\W^\eps(\phi^\eps_k,\phi^\eps_l)(x,\xi)=\W^\eps(\psi^\eps_k,\psi^\eps_l)(z(x,\xi)),
$$ 
and the claimed result follows by Corollary~\ref{cor:wigner}.
\end{proof}


\appendix
\section{Polar decomposition}
\label{app:pol}

The following observations are not needed for computing the phase space transforms of the Hagedorn wavepackets. However, they shed further light on the relation to the Hermite polynomials in the general case of invertible $Q\in\C^{d\times d}$. 

\begin{lemma}
Let $\eps>0$, $q,p\in\R^{d}$, and $Q,P\in\C^{d\times d}$ be matrices satisfying (\ref{eq:QP}). 
If we decompose $Q=|Q|U^*$ with $|Q|=(QQ^*)^{1/2}$ and $U\in\C^{d\times d}$ a unitary matrix, then $PQ^{-1}=(PU)|Q|^{-1}$, and the pair $|Q|$, $PU$ satisfies condition~(\ref{eq:QP}). Moreover, 
\begin{equation}
\label{eq:AU}
A^\dagger[q,p,Q,P] = U A^\dagger[q,p,|Q|,PU].
\end{equation}
\end{lemma}

\begin{proof}
We observe
\begin{eqnarray*}
0 
&=& Q^TP-P^TQ = (|Q|U^*)^TP-P^T|Q|U^*\\
&=& U^{-T}(|Q|^T(PU)-(PU)^T|Q|)U^*,\\
2i\Id 
&=& Q^*P-P^*Q=(|Q|U^*)^*P-P^*|Q|U^*\\
&=&U(|Q|^*(PU)-(PU)^*|Q|)U^*.
\end{eqnarray*}
Therefore, $|Q|^T(PU)-(PU)^T|Q|=0$ and $|Q|^*(PU)-(PU)^*|Q|= 2i\Id$. For the ladder operators we compute
\begin{eqnarray*}
A^\dagger[q,p,|Q|,PU] &=& 
\tfrac{i}{\sqrt{2\eps}} ((PU)^* \op_\eps(x-q) - |Q|^*\op_\eps(\xi - p))\\
&=&
\tfrac{i}{\sqrt{2\eps}} U^*(P^* \op_\eps(x-q) - Q^*\op_\eps(\xi - p))\\
&=& U^* A^\dagger[q,p,Q,P].
\end{eqnarray*}
\end{proof} 

The unitary relation (\ref{eq:AU}) of the ladder operators provides information on the Hagedorn wavepackets associated with indices of the same modulus: We first enumerate the multi-indices $k\in\N^d$ of equal modulus redundantly by setting $\tilde\nu_0=(0,\ldots,0)\in\N^d$,
$$
\tilde\nu_{|k|+1} = {\rm vec}\begin{pmatrix}
\tilde\nu_{|k|,1}+e_1       & \cdots & \tilde\nu_{|k|,1}+e_d\\
\vdots                      & \ddots & \vdots\\
\tilde\nu_{|k|,d^{|k|}}+e_1 & \cdots & \tilde\nu_{|k|,d^{|k|}}+e_d
\end{pmatrix},\qquad |k|\in\N,
$$
such that $\tilde\nu_{|k|+1}$ is a vector of length $d^{|k|+1}$, whose entries are multi-indices in $\N^d$. 
Then, we mark repeated occurrences of multi-inidizes by setting
$$
\nu_{|k|,j} = \left\{
\begin{array}{ll} 
\infty & \exists j'<j:\, \tilde\nu_{|k|,j'}=\tilde\nu_{|k|,j},\\
\tilde\nu_{|k|,j} & \mbox{otherwise,}
\end{array}
\right.
$$ 
for $j=1,\ldots,d^{|k|}$. This redundant book-keeping allows to reformulate the creation process (\ref{eq:crea}) on the level of the $|k|$th eigenspace as $\vec\phi_0^\eps=\phi_0^\eps$, 
$$
\vec\phi_{|k|+1}^\eps
= {\rm vec}
\begin{pmatrix}
\frac{1}{\sqrt{(\nu_{|k|,1})_1+1}} A_1^\dagger \vec\phi_{|k|,1}^\eps & \cdots & 
\frac{1}{\sqrt{(\nu_{|k|,1})_d+1}} A_d^\dagger \vec\phi_{|k|,1}^\eps\\
\vdots & \ddots & \vdots\\
\frac{1}{\sqrt{(\nu_{|k|,d^{|k|}})_1+1}} A_1^\dagger \vec\phi_{|k|,d^{|k|}}^\eps & \cdots & 
\frac{1}{\sqrt{(\nu_{|k|,d^{|k|}})_d+1}} A_d^\dagger \vec\phi_{|k|,d^{|k|}}^\eps
\end{pmatrix}.
$$
We note that the normalization with $1/\sqrt{(\nu_{|k|,j})_l+1}$ produces a zero whenever a multi-index is repeated in the enumeration.

\begin{proposition}
\label{prop:kron}
Let $\eps>0$, $q,p\in\R^d$, and $Q,P\in\C^{d\times d}$ satisfy (\ref{eq:QP}). We decompose $Q=|Q|U^*$ with unitary $U\in\C^{d\times d}$. Then, $|Q|=|Q|^T\in\R^{d\times d}$, and 
\begin{equation}
\label{eq:kron}
\vec\phi^\eps_{|k|}[q,p,Q,P](x) = U^{|k|\otimes}\vec\phi^\eps_{|k|}[q,p,|Q|,PU],
\end{equation}
where $U^{\otimes|k|}=U\otimes\cdots\otimes U\in\C^{d^{|k|}\times d^{|k|}}$ denotes the $|k|$-fold Kronecker product of $U$ with itself. 
\end{proposition}

\begin{proof}
We set 
\begin{eqnarray*}
A^\dagger &=& A^\dagger[q,p,Q,P],\qquad\quad\; \vec\phi_{|k|}^\eps = \vec\phi_{|k|}^\eps[q,p,Q,P],\\ 
D^\dagger &=& A^\dagger[q,p,|Q|,PU],\qquad\vec\psi_{|k|}^\eps = \vec\phi_{|k|}^\eps[q,p,|Q|,PU].
\end{eqnarray*}
By the relation (\ref{eq:AU}), $A_j^\dagger = (U D^\dagger)_j = u_j^T D^\dagger$
with $u_1,\ldots,u_d\in\C^d$ the row vectors of $U$. Moreover, for arbitrary $W\in\C^{m\times m}$, $w\in\C^{m}$, 
and $a\in\C^d$  
\begin{eqnarray*}
\begin{pmatrix}(u_1^T a )Ww\\ \vdots\\ (u_d^T a) Ww\end{pmatrix} 
&=&
\begin{pmatrix} 
(u_{11}a_1) Ww + \cdots + (u_{1d}a_d) Ww\\ \vdots\\
(u_{d1}a_1) Ww + \cdots + (u_{dd}a_d) Ww
\end{pmatrix}\\
&=& 
\begin{pmatrix}
u_{11}W & \cdots & u_{1d}W\\ \vdots & \ddots & \vdots\\ u_{d1}W & \cdots & u_{dd}W
\end{pmatrix} 
\begin{pmatrix}a_1 w\\ \vdots \\a_d w\end{pmatrix} 
= U\otimes W \begin{pmatrix}a_1 w\\ \vdots \\a_d w\end{pmatrix}.
\end{eqnarray*}
Assuming that the claimed identity (\ref{eq:kron}) holds for $|k|$, we therefore obtain
\begin{eqnarray*}
\vec\phi_{|k|+1}^\eps
&=&
\begin{pmatrix}
\frac{1}{\sqrt{(\nu_{|k|,1})_1+1}}\,A_1^\dagger\, \vec\phi_{|k|,1}^\eps\\\vdots\\ 
\frac{1}{\sqrt{(\nu_{|k|,d^{|k|}})_d+1}}\,A_d^\dagger\, \vec\phi_{|k|,d^{|k|}}^\eps
\end{pmatrix}
= 
\begin{pmatrix}
(u_1^T D^\dagger)\, U^{\otimes|k|}\, \vec\psi_{|k|}^\eps\\\vdots\\ 
(u_d^T D^\dagger)\, U^{\otimes|k|}\, \vec\psi_{|k|}^\eps
\end{pmatrix}\\*[1ex]
&=&
U\otimes U^{\otimes|k|} 
\begin{pmatrix}
\frac{1}{\sqrt{(\nu_{|k|,1})_1+1}}\,D_1^\dagger\, \vec\psi_{|k|,1}^\eps\\ \vdots \\
\frac{1}{\sqrt{(\nu_{|k|,d^{|k|}})_d+1}}\,D_d^\dagger\, \vec\psi_{|k|,d^{|k|}}^\eps
\end{pmatrix}
=
U^{\otimes(|k|+1)} \vec\psi_{|k|+1}^\eps.
\end{eqnarray*}
\end{proof}

\section{Weyl quantization}
\label{app:we}

The raising and lowering operators $A^\dagger$ and $A$ as well as the generalized harmonic oscillator $\tfrac12\left(A^T A^\dagger + (A^\dagger)^T A \right)$ of the Hagedorn wavepackets can be viewed as Weyl quantized operators obtained from smooth phase space functions of subquadratic growth $a:\R^{2d}\to\C^d$, 
$$
(\op_\eps(a)\phi)(x) = (2\pi\eps)^{-d} \int_{\R^{2d}} a(\tfrac12(x+y),\xi) \e^{i(x-y)^T\xi/\eps}\phi(y) dy d\xi
$$
for Schwartz functions $\phi:\R^d\to\C$. The emerging phase space function 
$$
z(x,\xi)=-i\left(P^T(x-q)-Q^T(\xi-p)\right)
$$ 
generalizes the complex number $z(x,\xi)=x+i\xi$, which characterizes the Wigner and FBI transform of the Hermite functions. It also appears in the Wigner function of the Hagedorn wavepackets.  

\begin{lemma}
\label{lem:op}  
Let $\eps>0$, $q,p\in\R^{d}$, and $Q,P\in\C^{d\times d}$ be matrices satisfying (\ref{eq:QP}). Let  $A^\dagger=A^\dagger[q,p,Q,P]$ and $A=A[q,p,Q,P]$. Then,
$$
A^\dagger = \tfrac{1}{\sqrt{2\eps}}\op_\eps(\overline z),\qquad 
A = \tfrac{1}{\sqrt{2\eps}}\op_\eps(z),\qquad
\tfrac12 \sum_{j=1}^d(A_j A_j^\dagger + A_j^\dagger A_j)=\tfrac{1}{2\eps}\op_\eps(|z|^2)
$$  
with $z(x,\xi) = -i\left(P^T(x-q)-Q^T(\xi-p)\right)$ for $x,\xi\in\R^d$
and $|z|^2 = \overline z^T z$.
\end{lemma}

\begin{proof}
We observe 
$$
|z|^2 = (x-q)^T|P|^2(x-q)+(\xi-p)^T|Q|^2(\xi-p)-(x-q)^T M(\xi-p)
$$
with $M=PQ^*+\overline{P}Q^T\in\R^{d\times d}$. Moreover,
\begin{eqnarray*}
\lefteqn{\tfrac12\left(A^T A^\dagger + (A^\dagger)^T A \right) =}\\
&& 
\tfrac{1}{2\eps} \left( 
\op_\eps(x-q)^T |P|^2 \op_\eps(x-q) + \op_\eps(\xi-p)^T |Q|^2 \op_\eps(\xi-p)\right.\\
&& 
\left. - \op_\eps(x-q)^T M\op_\eps(\xi-p) - (M\op_\eps(\xi-p))^T \op_\eps(x-q)
\right)
\end{eqnarray*}
and
\begin{eqnarray*}
\op_\eps(x-q)^T M\op_\eps(\xi-p) &=& \op_\eps((x-q)^T M(\xi-p)) - \eps\tfrac{i}{2} \tr M,\\ 
(M\op_\eps(\xi-p))^T \op_\eps(x-q) &=& \op_\eps((x-q)^T M(\xi-p)) + \eps\tfrac{i}{2} \tr M.
\end{eqnarray*}
\end{proof}

%
%
%

\subsection*{Acknowledgments.} We thank the anonymous referees for pointing us to 
the generalized coherent states. We also thank Matthias N\"utzel and 
Ilja Klebanov for their careful reading of the manuscript.


\end{document}